\documentclass[times]{elsarticle}

\usepackage{amsmath}

\newtheorem{theorem}{Theorem}

\newtheorem{lemma}{Lemma}
\newtheorem{corollary}{Corollary}
\newtheorem{observation}{Observation}
\newcommand{\eop}{\hspace*{\fill}$\Box$}
\newproof{pf}{Proof}
\newenvironment{proof}{\begin{pf}}{\eop\end{pf}}

\newcommand{\GridLabel}{M}
\newcommand{\Grid}[1]{\GridLabel(#1)}

\newcommand{\bpi}{\bar\pi}
\newcommand{\gM}{\mu}
\newcommand{\lmax}{\ell_{\text{\upshape max}}}
\newcommand{\lmin}{\ell_{\text{\upshape min}}}
\newcommand{\V}[1]{V^{#1}}
\newcommand{\CF}[2]{W_{#1}^{#2}}
\newcommand{\CS}[3]{S_{#1}^{#2,#3}}

\begin{document}

\begin{frontmatter}

\title{Separator-Based Graph Embedding into Multidimensional Grids
 with Small Edge-Congestion}
\author{Akira Matsubayashi\corref{cor1}}
\ead{mbayashi@t.kanazawa-u.ac.jp}
\address{Division of Electrical Engineering and Computer Science,
 Kanazawa University, Kanazawa 920-1192, Japan}
\cortext[cor1]{Tel/Fax: +81 76 234 4837}

\begin{abstract}
We study the problem of embedding a guest graph
 with minimum edge-congestion into a multidimensional grid
 with the same size as that of the guest graph.
Based on a well-known notion of graph separators,
 we show that an embedding with a smaller edge-congestion can be 
 obtained if the guest graph has a smaller separator,
 and if the host grid has a higher but constant dimension.
Specifically, we prove that
 any graph with $N$ nodes, maximum node degree $\Delta$,
 and with a node-separator of size $O(n^\alpha)$ ($0\leq\alpha<1$)
 can be embedded into
 a grid of
 a fixed dimension $d\geq 2$ with
 at least $N$ nodes,
 with an edge-congestion of
 $O(\Delta)$ if $d>1/(1-\alpha)$,
 $O(\Delta\log N)$ if $d=1/(1-\alpha)$, and
 $O(\Delta N^{\alpha-1+\frac{1}{d}})$ if $d< 1/(1-\alpha)$.
This edge-congestion achieves constant ratio approximation
 if $d>1/(1-\alpha)$, and
 matches
 an existential lower bound within a constant factor
 if $d\leq 1/(1-\alpha)$.
Our result implies
 that if the guest graph has an excluded minor of a fixed size,
 such as a planar graph,
 then we can obtain an edge-congestion of
 $O(\Delta\log N)$ for $d=2$ and $O(\Delta)$ for any fixed $d\geq 3$.
Moreover, 
 if the guest graph has a fixed treewidth,
 such as a tree, an outerplanar graph, and a series-parallel graph,
 then we can obtain an edge-congestion of
 $O(\Delta)$ for any fixed $d\geq 2$.
To design our embedding algorithm,
 we introduce \emph{edge-separators bounding expansion}, such that
 in partitioning a graph into isolated nodes using edge-separators recursively,
 the number of outgoing edges from a subgraph to be partitioned in a recursive
 step is bounded.
We present an algorithm to construct
 an edge-separator with expansion of $O(\Delta n^\alpha)$
 from a node-separator of size~$O(n^\alpha)$.
\end{abstract}


\end{frontmatter}

\section{Introduction}
The \emph{graph embedding} of a guest graph into a host graph is
 to map (typically one-to-one) nodes and edges of the guest graph
 onto nodes and paths of the host graph, respectively,
 so that an edge of the guest graph is mapped onto a path
 connecting the images of end-nodes of the edge.
The graph embedding problem is
 to embed a guest graph into a host graph with certain
 constraints and/or optimization criteria.
This problem has applications such as
 efficient VLSI layout and parallel computation.
I.e., the problem of efficiently laying out VLSI can be formulated as
 the graph embedding problem with modeling a design rule on wafers
 and a circuit to be laid out as host and guest graphs, respectively.
Also, the problem of efficiently implementing a parallel algorithm
 on a message passing parallel computer system consisting
 of processing elements connected by an interconnection network
 can be formulated as the graph embedding problem with
 modeling the interconnection network and
 interprocess communication in the parallel algorithm
 as host and guest graphs, respectively.
See for a survey, e.g., \cite{RH01}.
The major criteria to measure the efficiency of an embedding
 are dilation, node-congestion, and edge-congestion.
In this paper, we consider the problem of embedding a guest graph
 with the minimum edge-congestion
 into a $d$-dimensional grid
 with $d\geq 2$ and the same size as that of the guest graph.
Embeddings into grids with the minimum edge-congestion are important
 for both VLSI layout and parallel computation.
Actually, design rules on wafers in VLSI
 are usually modeled as $2$-dimensional grids,
 and an edge-congestion provides a lower bound on the number of layers
 needed to lay out a given circuit.
As for parallel computation, multidimensional grid networks, including
 hypercubes, are popular for
 interconnection networks.
On interconnection networks adopting
 circuit switching or wormhole routing, in particular,
 embeddings with the edge-congestion of $1$ are essential
 to minimize the communication latency \cite{KL91,LC97,TT00}.
In addition, the setting that host and guest graphs have the same number of
 nodes is important for parallel computation because
 the processing elements are expensive resource and
 idling some of them is wasteful.

\subsection*{Previous Results}
Graph embedding into grids with small edge-congestion
 has extensively been studied.
Table~\ref{tb:summary} summarizes previous results of
 graph embeddings minimizing edge-congestion
 (and other criteria as well in some results) for various combinations of
 guest graphs and host grids.

\begin{table}
\caption{Previous results of graph embeddings minimizing edge-congestion.}
\label{tb:summary}
\begin{center}
\footnotesize
\begin{tabular}{lccccc}
\hline
Guest Graph & \multicolumn{2}{c}{Host Grid} & Congestion & Dilation &\\
$N$: \# nodes, $\Delta$: max degree
 & \# nodes & dimension & & &\\
$s$: separator size & & & & &\\
\hline
\hline
connected planar graph & $N$ & 2 & NP-hard for 1 & any &\cite{FW91}\\
connected graph & $2^{\lceil\log_2N\rceil}$ & $\lceil\log_2N\rceil$
 & NP-hard for 1 & any &\cite{KL91}\\
complete binary tree & $N+1$ &  2 & 2 & $O(\sqrt N)$ &\cite{Zi91}\\
complete binary tree & $N+1$ & 4 & 1 & $O(N^{1/4})$&\cite{LC97}\\
complete binary tree & $N+O(\sqrt{N})$ & 2 & 1 & $O(\sqrt N)$&\cite{Go87}\\
complete $k$-ary tree ($k\geq 3$) & $N+O(N/\sqrt{k})$ & 2
 & $\lceil k/2\rceil+1$ & $O(\sqrt N)$&\cite{TT00}\\
binary tree & $2^{\lceil\log_2N\rceil}$ & $\lceil\log_2N\rceil$
 & 5 & $\lceil\log_2 N\rceil$&\cite{MU99}\\
2-D $h\times w$-grid ($h\leq w$) & $h'w'\geq N^\ast$ & 2
 & $\lceil h/h'\rceil+1$ & $\lceil h/h'\rceil+1$ &\cite{RS01}\\
2-D $h\times w$-grid ($h\leq w$) & $h'w'\geq N^\dag$ & 2 & 5 & 5 &\cite{RS01}\\
2-D $h\times w$-grid ($h\leq w$) & $h'w'\geq N^\dag$ & 2 & 4 & $\geq 4h-3$
 &\cite{RS01}\\
2-D grid & $2^{\lceil\log_2N\rceil}$ & $\lceil\log_2N\rceil$
 & 2 & 3 &\cite{RS98}\\
$\Delta\leq 4$, $s=O(n^\alpha)$, $\alpha<1/2$ & $O(N)$ & 2 & 1
 & $O(\sqrt N/\log N)$ &\cite{BL84}\\
$\Delta\leq 4$, $s=O(\sqrt{n})$ & $O(N\log^2N)$ & 2 & 1
 & $O(\frac{\sqrt N\log N}{\log\log N})$ &\cite{BL84}\\
$\Delta\leq 4$, $s=O(n^\alpha)$, $\alpha>1/2$
 & $O(N^{2\alpha})$ & 2 & 1
 & $O(N^\alpha)$ &\cite{BL84}\\
tree width $t$ & $2^{\lceil\log_2N\rceil}$ & $\lceil\log_2N\rceil$
 & $O(\Delta^4t^3)$ & $O(\log(\Delta t))$&\cite{HM02-1}\\
 $s=\log^{O(1)}N$ & $2^{\lceil\log_2N\rceil}$
 & $\lceil\log_2N\rceil$ & $\Delta^{O(1)}$ & $O(\log\Delta)$&\cite{HM03}\\
$\Delta=O(1)$ & $N$ & $d=O(1)$
 & $O(N^{1/d}\log N)$ & $O(N^{1/d}\log N)$ &\cite{LR99}\\
$\Delta\leq 2\lceil\log_2N\rceil$
 & $2^{2\lceil\log_2N\rceil}$ & $2\lceil\log_2N\rceil$
 & 1 & $2\lceil\log_2N\rceil$ &\cite{MU99}\\
\hline
\end{tabular}\\
$\mbox{}^\ast$ $h'\times w'$-grid with $h'<h\leq w< w'$\\
$\mbox{}^\dag$ $h'\times w'$-grid with $h<h'\leq w'< w$
\end{center}
\end{table}

VLSI layout has been studied through formulating the layout as
 the graph embedding into a $2$-dimensional grid
 with objective of minimizing the grid under
 constrained congestion-$1$ routing \cite{Ul84}.
Leiserson \cite{Le83} and Valiant \cite{Va81} independently proposed
 such embeddings based on graph separators.
In particular, 
 it was proved in \cite{Le83} that
 any $N$-node graph with maximum node degree at most $4$
 and an edge-separator of size $O(n^\alpha)$ can be laid out
 in an area of $O(N)$ if $\alpha<1/2$,
 $O(N\log^2 N)$ if $\alpha=1/2$,
 and $O(N^{2\alpha})$ if $\alpha>1/2$.
A separator of a graph $G$ is a set $S$ of either nodes or edges
 whose removal partitions the node set $V(G)$ of $G$ into two subsets
 of roughly the same size with no edge between the subsets.
The graph $G$ is said to have a (recursive) separator of size $s(n)$ if
 $|S|\leq s(|V(G)|)$ and the subgraphs partitioned by $S$ recursively
 have separators of size $s(n)$.
Separators are important tools to design divide-and-conquer algorithms
 and have been extensively studied.
Bhatt and Leighton \cite{BL84} achieved a better layout with
 several nice properties including reduced dilation as well as
 the same or better area as that of \cite{Le83}
 by introducing a special type of edge-separators called \emph{bifurcators}.
An approximation algorithm for VLSI layout was proposed in \cite{EGS03}.
Separator-based graph embeddings on hypercubes were presented
 in \cite{BCLR92,Ob94,HM03}.
In particular, Heun and Mayr \cite{HM03} proved that
 any $N$-node graph with
 maximum node degree $\Delta$ and
 an extended edge-bisector of polylogarithmic size 
 can be embedded into
 a $\lceil\log_2 N\rceil$-dimensional cube with
 a dilation of $O(\log\Delta)$
 and an edge-congestion of $\Delta^{O(1)}$.

A quite general embedding based on the multicommodity flow
 was presented by Leighton and Rao~\cite{LR99}, who proved
 that any $N$-node bounded degree graph $G$ can be embedded into
 an $N$-node bounded degree graph $H$
 with both dilation and edge-congestion of $O((\log N)/\alpha)$,
 where $\alpha$ is the \emph{flux} of $H$, i.e.,
$\min_{U\subset V(H)}\frac{
|\{(u,v)\in E(H)\mid u\in U,\ v\in V(H)\setminus U\}|}{
\min\{|U|,|V(H)\setminus U|\}}
$.
This implies that
 $G$ can be embedded into an $N$-node $d$-dimensional grid with
 both dilation and edge-congestion of $O(N^{1/d}\log N)$
 for any fixed $d$.

\subsection*{Contributions and Technical Overview}
In this paper,
 we improve previous graph embeddings into grids and hypercubes
 in terms of
 edge-congestion, arbitrary dimension, and
 minimum size of host grids.
In particular, we claim that if a guest graph has a small separator, then
 we do not need grids with large dimension, such as hypercubes, to
 suppress the edge-congestion.

First, we present an embedding algorithm based on the permutation routing.
The permutation routing is to
 construct paths connecting given pairs of source and destination nodes
 such that
 no two pairs have the same sources or the same destinations.
This embedding algorithm achieves an edge-congestion as stated in the
 following theorem:

\begin{theorem}
\label{th:PermutationEmbedding}
Any graph with $N$ nodes and maximum node degree $\Delta$
 can be embedded into a $d$-dimensional
 $\ell_1\times\cdots\times \ell_d$-grid ($\prod_{i=1}^d \ell_i\geq N$) with
 a dilation at most $2\sum_{i=1}^{d}\ell_i$ and
 an edge-congestion at most $2\lceil\Delta/2\rceil\cdot\max_{i}\{\ell_i\}$.
\end{theorem}
We prove this theorem in Sect.~\ref{ssc:PermutationRoutingEmbedding}
 by observing that
 for any one-to-one mapping of nodes of a guest graph $G$
 to nodes of a host graph $H$,
 routing edges of $G$ on $H$ can be reduced to at most $\lceil\Delta/2\rceil$
 instances of permutation routing, and that
 the permutation routing algorithm proposed in \cite{BA91}
 has an edge-congestion at most $2\cdot\max_i\{\ell_i\}$.
 Theorem~\ref{th:PermutationEmbedding} achieves
 an edge-congestion
 of $2\lceil\Delta/2\rceil\lceil N^{1/d}\rceil$ if
 $\ell_i=\lceil N^{1/d}\rceil$ for each $i$.
It is worth noting that this edge-congestion can slightly be improved 
 if the host grid $H$ is a $d$-dimensional cube.
It is well-known that
 any one-to-one mapping of $2^{d+1}$ inputs to $2^{d+1}$ outputs on
 a $d$-dimensional Bene\v{s} network
 can be routed with the edge-congestion $1$ \cite{Be65}.
We can easily observe that
 mapping the nodes in each row of the Bene\v{s} network
 to each node of $H$
 induces a (many-to-one) embedding with the edge-congestion $4$.
Because
 each node of $H$ has exactly two inputs and two outputs in a row
 of the Bene\v{s} network,
 any pair of instances of permutation routing on $H$
 can be routed with an edge-congestion at most $4$.
At most $\lceil\Delta/2\rceil$
 instances of permutation routing, obtained from
 any one-to-one mapping of nodes of $G$ to nodes of $H$ and
 from edges of $G$, can be grouped into
 $\lceil\lceil\Delta/2\rceil/2\rceil=\lceil\Delta/4\rceil$ pairs of instances
 of permutation routing.
Therefore,
 $G$ can be embedded into
 a $\lceil\log_2 N\rceil$-dimensional cube with an edge-congestion at most
 $4\lceil\Delta/4\rceil$.

Second, we present an embedding algorithm based on separators that
 achieves an edge-congestion as stated in the following theorem:

\begin{theorem}
\label{th:EmbedGeneral}
Suppose that $G$ is a graph with $N$ nodes,
 maximum node degree $\Delta$,
 and with a node-separator of size $O(n^\alpha)$
 ($0\leq\alpha< 1$), and that
 $\GridLabel$ is a grid with a fixed dimension $d\geq 2$,
 at least $N$ nodes, and with constant aspect ratio.
Then,
 $G$ can be embedded into $\GridLabel$ with
 a dilation of $O(dN^{1/d})$, and with
 an edge-congestion of
 $O(\Delta)$ if $d>1/(1-\alpha)$,
 $O(\Delta\log N)$ if $d=1/(1-\alpha)$, and
 $O(\Delta N^{\alpha-1+\frac{1}{d}})$  if $d<1/(1-\alpha)$.
\end{theorem}

The basic idea of Theorem~\ref{th:EmbedGeneral} is to partition
 the guest graph and the host grid using their edge-separators,
 embed the partitioned guest graphs into the partitioned host grids recursively,
 and to route cut edges of the guest graph on the host grid.
We use Theorem~\ref{th:PermutationEmbedding} to route cut edges
 with a nearly minimum edge-congestion in each recursive step.
However,
 just doing this is not sufficient for our goal.
In fact, we need further techniques
 to suppress the total edge-congestion incurred
 by whole recursive steps from upper to lower levels.
There are two reasons of the insufficiency.

The first reason is that
 recursive steps from upper to lower levels
 may use the same edge of the grid, which yields
 an edge-congestion of $\Omega(\log N)$
 if we minimize the edge-congestion only in each individual recursive step.
This is a crucial barrier to achieve
 an edge-congestion of $O(\Delta)$ for $d> 1/(1-\alpha)$.
To solve this, we divide the edge set of the grid into
 $\Theta(\log N)$ subsets of appropriate size and use each subset only in
 a constant number of recursive steps.

The second and more significant reason is that
 a small subgraph of the guest graph to be embedded in a lower recursive step
 may have nodes incident to quite a large number of edges that have been
 cut in upper levels, which yields a large edge-congestion.
Specifically, if such a subgraph has $n$ nodes and $x$ outgoing edges
 to the other part of the guest graph,
 then because a subgrid into which the subgraph is embedded has
 $O(dn^{1-\frac 1 d})$ outgoing edges,
 the edge-congestion is lower bounded by
 $x/O(dn^{1-\frac 1 d})=\Omega(xn^{\frac 1 d-1}/d)$.
A standard edge-separator aims to minimize the number of edges
 to be cut to partition a graph.
Thus, if we recursively use such edge-separators to partition a graph
 into small pieces,
 then
 although the number of cut edges in each recursive step is bounded,
 the number of outgoing edges from a subgraph to be embedded
 in a lower recursive step may become extremely large compared to
 the number of nodes of the subgraph.
Therefore,
 we introduce edge-separators bounding \emph{expansion}, i.e.,
 the number of outgoing edges from a subgraph in each recursive step, and
 present an algorithm to construct
 an edge-separator with expansion of $O(\Delta n^\alpha)$
 from a node-separator of size $O(n^\alpha)$.
We describe the algorithm for edge-separators with bounded expansion
 in Sect.~\ref{sc:Separators} and prove Theorem~\ref{th:EmbedGeneral}
 in Sect.~\ref{ssc:SeparatorBasedEmbedding}.

Theorem~\ref{th:EmbedGeneral}
 achieves constant ratio approximation
 for a fixed $d>1/(1-\alpha)$
 because any embedding has an edge-congestion at least $\Delta/(2d)$.
If $d\leq 1/(1-\alpha)$, then
 the edge-congestion of Theorem~\ref{th:EmbedGeneral} matches
 an existential lower bound within a constant factor.
The lower bound of $\Omega(\log N)$
 for $d=1/(1-\alpha)=2$ and $\Delta=O(1)$
 is derived from the following fact:
There exists an $N$-node guest graph with constant degree and a node-separator
 of size $O(\sqrt n)$
 whose any embedding into a $2$-dimensional grid with the edge-congestion~$1$
 requires $\Omega(N\log^2 N)$ nodes of the grid \cite{Le84}.\footnote{%
Strictly, this result is proved for the VLSI layout model.
However, we can easily generalize this result to the embedding model considered
 in this paper.
}
This implies that any embedding of the guest graph into a
 $2$-dimensional grid
 with $N$ nodes
 requires an edge-congestion of $\Omega(\log N)$.
This is because we can easily transform
 an embedding into an $N$-node grid with an edge-congestion $c$
 into
 another embedding into an $O(c^2N)$-node grid with the edge-congestion $1$
 by replacing each row and each column of the $N$-node grid
 with $O(c)$ rows and $O(c)$ columns, respectively.\footnote{%
It should be noted that the inverse transformation cannot be
 done in such a simple way.
In fact, we do not know whether or not the inverse transformation
 is always possible.
}
A similar transformation for VLSI layout is described in \cite{Ul84}.

The lower bound of $\Omega(\Delta N^{\alpha-1+\frac 1 d})$
 for $d<1/(1-\alpha)$ can be obtained as follows:
We consider a guest graph $G$ with
 $N$ nodes and
 a node-separator of size $n^\alpha$
 such that each node in a cut set $U\subseteq V(G)$
 with $|U|=N^\alpha$ is adjacent to every other node in $G$.
The graph $G$ obviously has $\Delta=N-1$.
Suppose that we arbitrarily divide $V(G)$ into two subsets of the same size.
Then, at least $(|U|/2)(N-1)/2=\Delta N^\alpha/4$ edges
 join nodes in one of the subsets and nodes in the other subset
 because
 at least half nodes of $U$ are contained in one of the subsets and
 adjacent to all nodes in the other subset.
On the other hand, we can divide a $d$-dimensional
 $N$-node grid
 into two subgrids of the same size by removing
 $O(N^{1-\frac 1 d})$ edges.
Thus, any embedding of $G$ into the grid has an edge-congestion at least
 $(\Delta N^\alpha/4)/O(N^{1-\frac 1 d})=\Omega(\Delta N^{\alpha-1+\frac 1 d})$.

Theorem~\ref{th:EmbedGeneral} has the following applications.
It is well-known that any planar graph has a node-separator
 of size $O(\sqrt{n})$ \cite{LT79}.
This was generalized in \cite{AST90b} so that
 any graph with an excluded minor of a fixed size has a node-separator of
 size $O(\sqrt{n})$.
Therefore, we obtain the following corollary:
\begin{corollary}
\label{cr:EmbedPlanar}
Any graph with $N$ nodes, maximum node degree $\Delta$, and
 with an excluded minor of a fixed size
 can be embedded into a grid of
 a fixed dimension $d$ with at least $N$ nodes
 and constant aspect ratio,
 with an edge-congestion of
 $O(\Delta\log N)$ for $d=2$ and $O(\Delta)$ for $d\geq 3$.
\end{corollary}

Graphs with a fixed treewidth,
 such as trees, outerplanar graphs, and series-parallel graphs
 have a node-separator of a fixed size \cite{Kl94}.
Therefore, we obtain the following corollary:
\begin{corollary}
\label{cr:EmbedSeriesParallel}
Any graph
 with $N$ nodes,
 maximum node degree $\Delta$,
 and with a fixed treewidth
 can be embedded into a grid of a fixed dimension at least $2$
 with at least $N$ nodes and constant aspect ratio,
 with an edge-congestion of $O(\Delta)$.
\end{corollary}

Our separator-based embedding algorithm performs in a polynomial time
 on the condition that a separator of the guest graph is given.
Although finding a separator of minimum size is generally
 NP-hard~\cite{GJS76,BJ92},
 approximation algorithms presented in \cite{LR99,FHL08,ARV09}
 can be applied to our algorithm.

All our embedding algorithms yield a dilation of order of the diameter
 of the host grid.
Although such a dilation is trivial when only the dilation is minimized,
 this is not the case when edge-congestion is minimized.
As we will demonstrate in Sect.~\ref{sc:Dilation}, in fact, 
 there exists an $N$-node guest graph
 whose any embedding
 with the edge-congestion $1$
 into an $N$-node
 $2$-dimensional grid requires a dilation of $\Theta(N)$,
 far from the diameter $\Theta(\sqrt N)$.
We do not know whether or not we can always achieve
 both a dilation of the host grid's diameter
 (even with a multiplicative constant factor)
 and constant ratio approximation for edge-congestion.
This is negative if the host graph is general.
As an example, suppose that $H$ is the host graph obtained from a
 complete binary tree with $N$ leaves
 by adding edges so that
 the leaves induce a $\sqrt N\times\sqrt N$-grid.
To be precise, 
 the $N/2^i$ leaves of a subtree rooted by a node
 at an even distance $i$ to the root
 induce a $\sqrt{N/2^i}\times\sqrt{N/2^i}$-subgrid.
If the guest graph $G$ is an $N$-node complete graph,
 then any embedding of $G$ into $H$ with a dilation of the diameter
 $O(\log N)$ of $H$ has an edge-congestion of $\Omega(N^2)$
 because $\Omega(N^2)$ edges of $G$ must be routed through
 a single node of the tree part in $H$ to achieve such a dilation, while
 $G$ can be embedded into the grid part in $H$ with a dilation of $O(\sqrt N)$
 and an edge-congestion of $O(N^{3/2})$ using a simple row-column routing.

\section{Preliminaries}
\label{sc:Preliminaries}
For a graph $G$, $V(G)$ and $E(G)$ are the node set and edge set
 of $G$, respectively.
We denote the set of integers $\{i\mid 1\leq i\leq\ell\}$ by $[\ell]$.
For a $d$-dimensional vector $v:=(x_i)_{i\in [d]}$,
 let $\pi_j(v):=x_j$ and $\bpi_j(v):=(x_i)_{i\in [d]\setminus\{j\}}$
 for $j\in [d]$.
We use $\pi_j$ and $\bpi_j$ also for a set of vectors and for a graph
 whose nodes are vectors.
I.e.,
 for a set $V$ of $d$-dimensional vectors,
 we denote
 $\{\pi_j(v)\mid v\in V\}$ and $\{\bpi_j(v)\mid v\in V\}$
 as $\pi_j(V)$ and $\bpi_j(V)$, respectively.
Moreover, for a graph $G$ with $V(G)=V$,
 we denote
 the graph with the node set
 $\bpi_j(V(G))$ and edge multiset
 $\{(\bpi_j(u),\bpi_j(v))\mid (u,v)\in E(G)\}$
 as $\bpi_j(G)$.
For positive integers $\ell_1,\ldots,\ell_d$,
 the \emph{$d$-dimensional $\ell_1\times\cdots\times \ell_d$-grid},
 denoted as $\Grid{\ell_i}_{i\in [d]}$,
 is a graph
 with the node set $\prod_{i\in [d]}[\ell_i]$, i.e.,
 the Cartesian product of sets $[l_1],\ldots,[l_d]$,
 and edge set
 $\{(u,v)\mid\exists j\in [d]\ \pi_j(u)=\pi_j(v)\pm 1, \bpi_j(u)=\bpi_j(v)\}$.
The \emph{aspect ratio} of $\Grid{\ell_i}_{i\in [d]}$ is
 $\max_{i,j\in [d]}\{\ell_j/\ell_i\}$.
An edge $(u,v)$ of $\Grid{\ell_i}_{i\in [d]}$ with
 $\pi_j(u)=\pi_j(v)\pm 1$
 is called a \emph{dimension-$j$ edge}.
The grid $\Grid{\ell_i}_{i\in[d]}$ is called the \emph{$d$-dimensional cube}
 if $\ell_i=2$ for every $i\in [d]$.

A \emph{routing request} on a graph $H$ is a pair of nodes,
  a \emph{source} and \emph{target}, of $H$.
A multiset of routing requests can be represented as
 a \emph{routing graph} $R$ with the node set $V(H)$ and directed edges
 joining the sources and targets of all the routing requests.
It should be noted that $R$ may have parallel edges and loops.
In particular, if $H$ is a $d$-dimensional grid,
 then
 $\bpi_j(R)$ is a routing graph
 with the multiset of edges $(\bpi_j(u),\bpi_j(v))$ for every $(u,v)\in E(R)$
 on the $(d-1)$-dimensional grid with
 node set $\bpi_j(V(H))$ (Fig~\ref{fig:RoutingGraph}).
$R$ is called a \emph{$p$-$q$ routing graph}
 if the maximum outdegree and indegree of $R$
 are at most $p$ and $q$, respectively.
A $1$-$1$ routing graph is also called a
 \emph{permutation routing graph}.
We define a \emph{routing} of $R$ as
 a mapping $\rho$ that maps each edge $(u,v)\in E(R)$ onto
 a set of edges of $H$ inducing a path connecting $u$ and $v$.
We denote $\rho((u,v))$ simply as $\rho(u,v)$.
The \emph{dilation} and \emph{edge-congestion} of
 $\rho$ are $\max_{e\in E(R)}|\rho(e)|$ and
$\max_{e'\in E(H)}|\{e\in E(R)\mid e'\in\rho(e)\}|$, respectively.

\begin{figure}
\begin{center}
\includegraphics[scale=1.0]{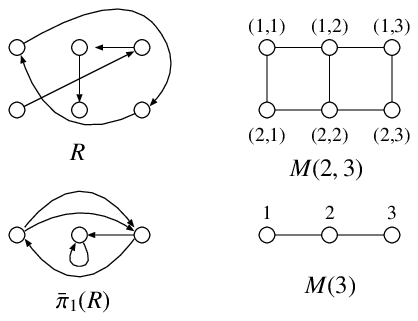}
\end{center}
\caption{A routing graph $R$ on $\Grid{2,3}$ and
 $\bpi_1(R)$ on $\Grid{3}$.}
\label{fig:RoutingGraph}
\end{figure}

An \emph{embedding} $\langle\phi,\rho\rangle$
 of a graph $G$ into a graph $H$ is a pair of mappings
 consisting of a one-to-one mapping $\phi:V(G)\rightarrow V(H)$
 and a routing $\rho$ of an arbitrary orientation of
 the graph with the node set $V(H)$ and
 edge set $\{(\phi(u),\phi(v))\mid (u,v)\in E(G)\}$.
The \emph{dilation} and \emph{edge-congestion} of
 the embedding $\langle\phi,\rho\rangle$ are defined as
 the dilation and edge-congestion of $\rho$, respectively.

\section{Edge-Separators with Bounded Expansion}
\label{sc:Separators}

The (recursive) node- and edge-separators are formally defined as follows:
Let $1/2\leq\beta<1$ and
 $s(n)$ be a non-decreasing function.
A graph $G$ has a \emph{$\beta$-node(edge, resp.)-separator
 of size $s(n)$} if $|V(G)|=1$, or
 if $G$ can be partitioned into two subgraphs with at most
 $\beta|V(G)|$ nodes
($\lceil\beta|V(G)|\rceil$ nodes, resp.)
 and with no edges connecting the subgraphs
 by removing at most $s(|V(G)|)$ nodes (edges, resp.),
 and the subgraphs recursively have a $\beta$-node(edge, resp.)-separator
 of size $s(n)$.
The process of partitioning $G$ into isolated nodes
 using the edge-separator repeatedly
 is often
 referred to 
 as a \emph{decomposition tree}.
The decomposition tree $\mathcal T$ 
 is a rooted tree having a set of subgraphs of $G$ as its node set
 $V(\mathcal T)$
 such that
 the root of $\mathcal T$ is $G$,
 each non-leaf node $H\in V(\mathcal T)$ has exactly two children
 obtained from $H$ by removing the edge-separator of $H$,
 and that
 each leaf node of $\mathcal T$ consists of a single node of $G$.
We call $\mathcal T$ a \emph{$\beta$-decomposition tree with expansion $x(n)$}
 if it can be constructed using a $\beta$-edge-separator, and
 for each $H\in V(\mathcal T)$, at most $x(|V(H)|)$ edges
 (called \emph{external edges} of $H$ in this paper) connect
 $V(H)$ and $V(G)\setminus V(H)$ (Fig.~\ref{fig:DecompositionTree}).

\begin{figure}
\begin{center}
\includegraphics[scale=1.0]{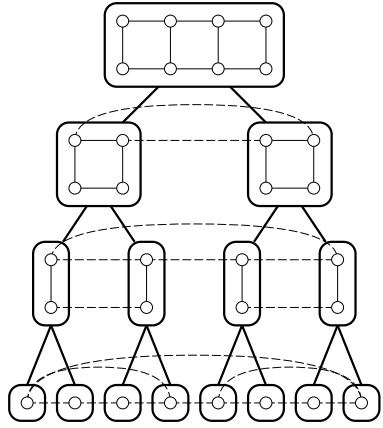}
\end{center}
\caption{%
A $(1/2)$-decomposition tree for $\Grid{2,4}$ with expansion $3$.
Dashed lines represent external edges of nodes of the decomposition tree.
}
\label{fig:DecompositionTree}
\end{figure}

A decomposition tree with reasonably small expansion
 can be obtained from
 a node-separator as stated in the following lemma:
\begin{lemma}
\label{lm:DecompositionTreeExpansion}
Any graph $G$ with maximum node degree $\Delta$ and a $\beta$-node-separator
 of size $Cn^\alpha$ ($C>0$, $0\leq\alpha<1$, $1/2\leq\beta<1$) has a
 $\frac\beta{1-\epsilon}$-decomposition tree with expansion
 $O(C\Delta n^\alpha/\epsilon)$, where $0<\epsilon<1-\beta$.
\end{lemma}
\begin{proof}
We present an algorithm constructing a desired decomposition tree $\mathcal T$.
We initially set $G$ as the root of $\mathcal T$ and
 construct $\mathcal T$ from the root toward leaves.
Assume that we have constructed $\mathcal T$ up to depth
 (distance to the root) $i-1\geq 0$.
For a subgraph $H$ of $G$ at depth $i-1$ in $\mathcal T$, we construct
 children $H_1$ and $H_2$ of $H$ as follows:
\begin{enumerate}
\item
We inductively assume the following:
\begin{enumerate}
\item
Each node of $\mathcal T$ up to depth $i-1$ has been constructed
 by partitioning a subgraph of $G$ using a $\beta$-node-separator
 and distributing the node-separator between the partitioned graphs.
Let $X_{i-1}$ be the set of nodes of $H$ contained in the node-separator
 used for any ancestor of $H$ in~$\mathcal T$.
\item
All the external edges of $H$ are incident to nodes in $X_{i-1}$.
\item
The graph $H'$ obtained from $H$ by removing $X_{i-1}$
 has a $\beta$-node-separator $S_i\subseteq V(H')$ of size $Cn^\alpha$.
\end{enumerate}
It should be noted that $X_0=\emptyset$, and therefore, these assumptions
 hold if $H=G$.
\item
\label{it:DecompositionTreeSeparation}
If $C|V(H')|^\alpha\leq\epsilon|V(H')|$, then we partition $H'$ into subgraphs
 $H'_1$ and $H'_2$ using the node-separator $S_i$
 with $|S_i|\leq C|V(H')|^\alpha$.
It follows that
$
|V(H'_1)|+|V(H'_2)|=|V(H')|-|S_i|\geq (1-\epsilon)|V(H')|
\geq\frac{1-\epsilon}\beta|V(H'_1)|.
$
Assume without loss of generality that $|V(H'_1)|\geq |V(H'_2)|$.
Then, there exists $1/2\leq\beta'\leq\frac\beta{1-\epsilon}$ with
 $|V(H'_1)|=\beta'(|V(H'_1)|+|V(H'_2)|)$.
\item
\label{it:DecompositionTreeNoSeparation}
If $C|V(H')|^\alpha>\epsilon|V(H')|$, then reset $S_i:=V(H')$, and
 arbitrarily choose
 $1/2\leq\beta'\leq\frac\beta{1-\epsilon}$.
\item
Partition $X_{i-1}\cup S_i$ into two disjoint sets $Y_1$ and $Y_2$ such that
 $|Y_1|=\lceil\beta'(|X_{i-1}|+|S_i|)\rceil$ and
 $|Y_2|=\lfloor(1-\beta')(|X_{i-1}|+|S_i|)\rfloor$.
\item
Let $H_j$ be the subgraph of $H$ induced by $V(H'_j)\cup Y_j$ for $j=1,2$.
We illustrate the construction in Fig.~\ref{fig:HPartition}.
\end{enumerate}

\begin{figure}
\begin{center}
\includegraphics[scale=1.0]{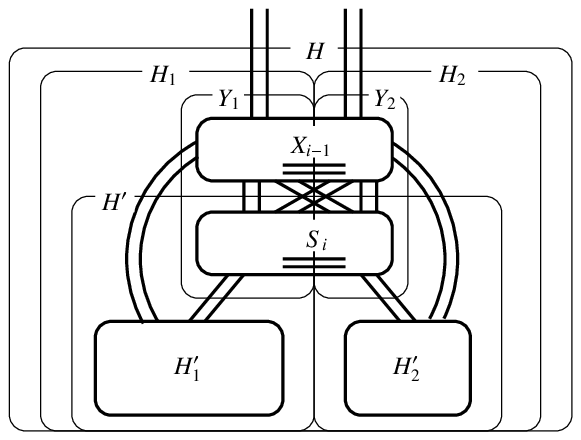}
\end{center}
\caption{Partition of $H$ into $H_1$ and $H_2$.
}
\label{fig:HPartition}
\end{figure}

We first observe that $H_1$ and $H_2$ satisfy the inductive assumptions of
 the algorithm.
For $j\in\{1,2\}$, by inductive assumption,
 $Y_j$ is the set of nodes of $H_j$ contained in
 the node-separator used for an ancestor of $H_j$.
As shown in Fig.~\ref{fig:HPartition}, all the external edges
 of $H_j$ are incident to nodes of $Y_j$.
Moreover, the subgraph of $H_j$ obtained by removing $Y_j$ is $H'_j$,
 which is the subgraph of $H'$ partitioned by the node-separator $S_i$ of~$H'$.
Therefore, $H'_j$ has a $\beta$-node-separator of size $Cn^\alpha$.

We then estimate the numbers of nodes of $H_1$ and $H_2$.
By definition, it follows that
\begin{align}
\label{eq:H1}
\begin{split}
|V(H_1)|&=|V(H'_1)|+|Y_1|
=\beta'(|V(H'_1)|+|V(H'_2)|)+\lceil\beta'(|X_{i-1}|+|S_i|)\rceil\\
&=\lceil\beta'|V(H)|\rceil,\ \text{and}
\end{split}\\
\label{eq:H2}
\begin{split}
|V(H_2)|&=|V(H'_2)|+|Y_2|
=(1-\beta')(|V(H'_1)|+|V(H'_2)|)+\lfloor(1-\beta')(|X_{i-1}|+|S_i|)\rfloor\\
&=\lfloor(1-\beta')|V(H)|\rfloor\leq\beta'|V(H)|.
\end{split}
\end{align}
These imply
 that the algorithm constructs $\mathcal T$ as
 a $\beta'$-decomposition tree of $G$.

We finally prove that for $j=1,2$, $H_j$ has
 $O(C\Delta|V(H_j)|^\alpha/\epsilon)$ external edges, implying
 expansion $O(C\Delta n^\alpha/\epsilon)$ of $\mathcal T$.
We prove this only for $H_1$ because the proof for $H_2$ is obtained with
 a similar argument.
Because all the external edges of $H_1$ are incident to $Y_1$,
 it suffices to show that $|Y_1|=O(C|V(H_1)|^\alpha/\epsilon)$.
When we construct children of $H_1$ using the algorithm, $X_i$ is set to $Y_1$.
Let
 $n_i:=|V(H_1)|$ and
 $n_j$ ($0\leq j<i$) be the number of nodes of the ancestor of
 $H_1$ at depth $j$ in $\mathcal T$.
Moreover, let $\beta_j$ ($1\leq j\leq i$) be $\beta'$ or $1-\beta'$ defined
 in Step \ref{it:DecompositionTreeSeparation}
 or \ref{it:DecompositionTreeNoSeparation}
 in partitioning the ancestor at depth $j-1$.
This implies that
 $n_j=\lceil\beta_jn_{j-1}\rceil$ as in (\ref{eq:H1}) or
 $n_j=\lfloor\beta_jn_{j-1}\rfloor$ as in (\ref{eq:H2}).
Therefore,
\begin{align}
\label{eq:nj_upperbound}
n_j
&\leq\lceil\beta_jn_{j-1}\rceil
\leq\beta_jn_{j-1}+1
\leq n_0\prod_{h=1}^j\beta_h+\sum_{\ell=1}^j\prod_{h=\ell+1}^{j}\beta_h
 =n_0\prod_{h=1}^j\beta_h+O(1),\ \text{and}\\
\label{eq:nj_lowerbound}
n_j&\geq\lfloor\beta_jn_{j-1}\rfloor\geq\beta_jn_{j-1}-1
\geq n_0\prod_{h=1}^j\beta_h-\sum_{\ell=1}^j\prod_{h=\ell+1}^{j}\beta_h
 =n_0\prod_{h=1}^j\beta_h-O(1).
\end{align}
Here, we have used the fact that
 $\sum_{\ell=1}^j\prod_{h=\ell+1}^{j}\beta_h
\leq\sum_{\ell=1}^j(\frac\beta{1-\epsilon})^{j-\ell}=O(1)$.
By the definition of $Y_1$,
 we have the following recurrence of $|X_i|$:
\[
|X_i|=|Y_1|=\lceil\beta_i(|X_{i-1}|+|S_i|)\rceil\leq\beta_i(|X_{i-1}|+|S_i|)+1
\leq\sum_{j=1}^i|S_j|\prod_{h=j}^i\beta_h+\sum_{j=1}^i\prod_{h=j+1}^{i}\beta_h.
\]
The number $|S_j|$ is less than $Cn_{j-1}^\alpha/\epsilon$ because
 $|S_j|\leq C|V(H')|^\alpha\leq Cn_{j-1}^\alpha<Cn_{j-1}^\alpha/\epsilon$
 if $S_j$ is defined in Step~\ref{it:DecompositionTreeSeparation}, and
 $|S_j|=|V(H')|<C|V(H')|^\alpha/\epsilon\leq Cn_{j-1}^\alpha/\epsilon$
 if $S_j$ is defined in Step~\ref{it:DecompositionTreeNoSeparation}.
Moreover, 
 $\sum_{j=1}^i\prod_{h=j+1}^{i}\beta_h=O(1)$ as estimated for
 (\ref{eq:nj_upperbound}) and~(\ref{eq:nj_lowerbound}).
Therefore, 
\[
\begin{split}
|Y_1|&<\sum_{j=1}^i\frac{C n_{j-1}^\alpha}\epsilon\prod_{h=j}^i\beta_h+O(1)
\leq\sum_{j=1}^i\frac C \epsilon\left(n_0\prod_{h=1}^{j-1}\beta_h
 +O(1)\right)^\alpha\prod_{h=j}^i\beta_h+O(1)
\quad\text{[by (\ref{eq:nj_upperbound})]}\\
&=O\left(\frac{Cn_0^\alpha}\epsilon\sum_{j=1}^i\prod_{h=1}^{j-1}\beta_h^\alpha
 \cdot\prod_{h=j}^i\beta_h\right)
=O\left(\frac{Cn_0^\alpha}\epsilon\prod_{h=1}^i\beta_h^\alpha\cdot\sum_{j=1}^i
 \prod_{h=j}^i\beta_h^{1-\alpha}\right)\\
&=O\left(\frac{C}\epsilon\left(n_0\prod_{h=1}^i\beta_h\right)^\alpha\cdot
\sum_{j=1}^i\left(\frac\beta{1-\epsilon}\right)^{(1-\alpha)(i-j+1)}\right)\\
&=O\left(\frac{C}\epsilon(n_i+O(1))^\alpha\cdot O(1)\right)
\quad\text{[by (\ref{eq:nj_lowerbound})]}\\
&=O\left(\frac{Cn_i^\alpha}\epsilon\right).\\
\end{split}
\]

Therefore, $\mathcal T$ is a desired decomposition tree.
\end{proof}

\section{Embedding Algorithm}
\label{sc:Embedding}
In this section, we first prove Theorem~\ref{th:PermutationEmbedding} by
 estimating the edge-congestion of
 the previously known permutation routing algorithm
 on multidimensional grids
 presented in~\cite{BA91}.
We then
 provide an embedding algorithm
 based on edge-separators with bounded expansion
 as well as the permutation routing algorithm.
Combining this algorithm with Lemma~\ref{lm:DecompositionTreeExpansion},
 we prove Theorem~\ref{th:EmbedGeneral}.
\subsection{Permutation Routing and Embedding}
\label{ssc:PermutationRoutingEmbedding}
Any permutation routing can be used to construct a graph embedding as follows:
\begin{lemma}
\label{lm:PermutationToEmbedding}
If any $1$-$1$ routing graph on a host graph $H$ can be routed with
 an edge-congestion at most $c$,
 then any graph with maximum node degree $\Delta$ can be embedded into
 $H$ with an edge-congestion at most $c\lceil\Delta/2\rceil$.
\end{lemma}
\begin{proof}
Let $G$ be a graph with maximum node degree $\Delta$
 to be embedded into $H$.
We arbitrarily choose
 a one-to-one mapping $\phi:V(G)\rightarrow V(H)$.
Let $G'$ be the graph with node set $\phi(V(G))$ and edge
 set $\{(\phi(u),\phi(v))\mid (u,v)\in E(G)\}$.
Because $G'$ is an undirected graph with maximum node degree $\Delta$,
 there is an orientation $R$ of $G'$ 
 whose maximum indegree and outdegree are both
 at most $\lceil\Delta/2\rceil$.
Such an orientation can be obtained
 by adding dummy edges joining nodes with odd degree so that
 the resulting graph has an Euler circuit,
 and by orienting edges along with the Euler circuit.
It suffices to prove that
 $R$ as a routing graph on $H$ can be routed with an edge-congestion
 at most $c\lceil\Delta/2\rceil$.

We decompose $R$ into at most
 $\lceil\Delta/2\rceil$ edge-disjoint $1$-$1$ routing graphs
 each of which has nodes $V(R)$ and edges with the same color
 in an edge-coloring of $R$ such that
 no two edges with the same
 sources or with the same targets have the same color.
Such coloring can be obtained by
 edge-coloring
 the bipartite graph consisting of the source and target sets of $R$, i.e.,
 two copies $V^+$ and $V^-$ of $V(R)$,
 and
 edges joining $u\in V^+$ and $v\in V^-$ for all $(u,v)\in E(R)$.
It should be noted that the resulting bipartite graph has node-degree
 at most $\lceil\Delta/2\rceil$, and hence,
 $\lceil\Delta/2\rceil$ colors are enough for the coloring.
Therefore, $R$ can be routed on $H$ with an edge-congestion
 at most $c\lceil\Delta/2\rceil$
 if each of the $1$-$1$ routing graphs
 can be routed with an edge-congestion at most $c$.
\end{proof}

The algorithm of \cite{BA91} routes
 a $1$-$1$ routing graph $R$ on $\GridLabel:=\Grid{\ell_i}_{i\in [d]}$
 as follows:
\begin{enumerate}
\item
\label{it:ColoringRoutingGraph}
Color edges of $R$ using at most $\ell_1$ colors so that
 when we identify edges in $R$ with corresponding edges in $\bpi_1(R)$,
 no two edges with the same sources or with the same targets in $\bpi_1(R)$
 have the same color.
This coloring can be obtained
 as done in the proof of Lemma~\ref{lm:PermutationToEmbedding}.
It should be noted that $\bpi_1(R)$ is a $\ell_1$-$\ell_1$ routing graph
 with node set $\bpi_1(V(\GridLabel))$.
\item
\label{it:DecompositionRoutingGraph}
Decompose $R$ into edge-disjoint subgraphs $R_1,\ldots,R_{\ell_1}$
 each of which has nodes $V(R)$ and edges with the same color.
\item
For each $i\in [\ell_1]$, 
 $\bpi_1(R_i)$ is a $1$-$1$ routing graph
 with node set $\bpi_1(V(\GridLabel))$.
Therefore, we can
 recursively find a routing $\rho_i$
 of $\bpi_1(R_i)$ on
 the $(d-1)$-dimensional subgrid
 $\GridLabel_i$ induced by the nodes $\{v\in V(\GridLabel)\mid\pi_1(v)=i\}$.
If $d=2$, i.e., if $\GridLabel_i$ is a path, then
 $\rho_i$ simply routes each routing request of $\bpi_1(R_i)$ on
 the path connecting its source and target in $\GridLabel_i$.
\item
We route each $(s,t)\in E(R_i)$ on the edge set consisting of
 dimension-$1$ edges connecting $s$ to $\GridLabel_i$,
 $\rho_i(\bpi_1(s),\bpi_1(t))$,
 and dimension-$1$ edges connecting $t$ to~$\GridLabel_i$.
\end{enumerate}

We can easily observe that in this algorithm,
 any dimension-$i$ edge of $\GridLabel$
 is contained in at most $2\ell_i$ images of $\rho$.
Moreover, each image of $\rho$ contains at most $2\ell_i$ dimension-$i$ edges.
I.e., $\rho$ has an edge-congestion of $2\cdot\max_{i\in [d]}\{\ell_i\}$
 and a dilation of $2\sum_{i=1}^d\ell_i$.
This property and Lemma~\ref{lm:PermutationToEmbedding} prove
 Theorem~\ref{th:PermutationEmbedding}.
With our aim of using this permutation routing algorithm to prove
 Theorem~\ref{th:EmbedGeneral},
 we generalize this property as the following lemma:
\begin{lemma}
\label{lm:GeneralRouting}
Let $R$ be a routing graph on
 $\GridLabel:=\Grid{\ell_i}_{i\in [d]}$ with $d\geq 2$ and
 $\ell_h:=\max_{i\in [d]}\{\ell_i\}$.
If $\bpi_h(R)$ is a $p$-$q$ routing graph
 with node set $\bpi_h(V(\GridLabel))$, then
 $R$ can be routed on $\GridLabel$
 with a dilation at most $2\sum_{i=1}^d \ell_i$ and
 an edge-congestion at most
 $2\cdot\max\{p,q\}$.
\end{lemma}
\begin{proof}
Assume without loss of generality that $h=1$ and $\ell_1\geq\cdots\geq \ell_d$.
We prove the lemma by induction on $d$.
If $d=2$, then $R$ has at most $\ell_2\cdot\max\{p,q\}$ edges.
We decompose $R$ into $\ell_1$ edge-disjoint subgraphs
 $R_1,\ldots,R_{\ell_1}$ so that
 $\bigcup_{i=1}^{\ell_1} E(R_i)=E(R)$ and
 $|E(R_i)|\leq\lceil\ell_2\cdot\max\{p,q\}/\ell_1\rceil\leq\max\{p,q\}$
 for $i\in [\ell_1]$.
For each $i\in [\ell_1]$,
 $\bpi_1(R_i)$ can be routed on the $1$-dimensional grid $\GridLabel_i$
 induced by the nodes $\{v\in V(\GridLabel)\mid \pi_1(v)=i\}$ with
 a dilation at most $\ell_2$ and
 an edge-congestion at most $\max\{p,q\}$.
The routing of $R$ is completed by adding
 the dimension-$1$ edges connecting $s$ to $\GridLabel_i$
 and $t$ to $\GridLabel_i$ for each $(s,t)\in E(R_i)$.
Any dimension-$1$ edge has a congestion at most $p+q$.
Moreover, at most $2\ell_1$ dimension-$1$ edges are added to each image of
 the routing.
Therefore, we have the lemma for $d=2$.

If $d\geq 3$, then
 we color edges of $R$ using at most $\ell_2\cdot\max\{p,q\}$ colors so that
 when we identify edges in $R$ with corresponding edges in $\bpi_2(\bpi_1(R))$,
 no two edges with the same sources or with the same targets
 in $\bpi_2(\bpi_1(R))$ have the same color.
Such coloring exists because $\bpi_2(\bpi_1(R))$ is
 a $\ell_2p$-$\ell_2q$ routing graph
 with node set $\bpi_2(\bpi_1(V(\GridLabel)))$.
Then,
 we decompose $R$ into $\ell_1$ edge-disjoint subgraphs
 $R_1,\ldots,R_{\ell_1}$ that have
 edge sets with disjoint sets of
 $\lceil\ell_2\cdot\max\{p,q\}/\ell_1\rceil\leq\max\{p,q\}$ colors.
This implies that
 $\bpi_2(\bpi_1(R_i))$ is a
 $\max\{p,q\}$-$\max\{p,q\}$ routing graph
 with node set $\bpi_2(\bpi_1(V(\GridLabel)))$
 for $i\in [\ell_1]$.
By induction hypothesis,
 $\bpi_1(R_i)$ can be routed on the $(d-1)$-dimensional subgrid 
 induced by the nodes $\{v\in V(\GridLabel)\mid \pi_1(v)=i\}$ with
 a dilation at most $2\sum_{i=2}^d\ell_i$ and
 an edge-congestion at most
 $2\cdot\max\{p,q\}$.
The routing of $R$ is completed by adding
 dimension-$1$ edges as done in the case of $d=2$, so that
 any dimension-$1$ edge has congestion at most $p+q$, and
 at most $2 \ell_1$ dimension-$1$ edges are added to each
 image of the routing.
Thus, we have routed $R$ with a dilation at most $2\sum_{i=1}^d\ell_i$ and
 an edge-congestion at most $2\cdot\max\{p,q\}$.
\end{proof}

If we do not have the assumption $\ell_h=\max_{i\in [d]}\{\ell_i\}$ in
 Lemma~\ref{lm:GeneralRouting}, then
 we can estimate
 $\lceil\ell_2\cdot\max\{p,q\}/\ell_1\rceil\leq
 \lceil\gM\cdot\max\{p,q\}\rceil$ in its proof, where
 $\gM$ is the aspect ratio of $\GridLabel$.
This means that
 $|E(R_i)|\leq\lceil\gM\cdot\max\{p,q\}\rceil$ for $d=2$, and that
 $\bpi_2(\bpi_1(R_i))$ is a
 $\lceil\gM\cdot\max\{p,q\}\rceil$-$\lceil\gM\cdot\max\{p,q\}\rceil$
 routing graph on $\bpi_2(\bpi_1(\GridLabel))$ for $d\geq 3$.
Therefore, initially assuming without loss of generality that
 $\ell_1=\ell_h$ and $\ell_2\geq\cdots\geq\ell_d$ in the proof,
 we have the following lemma:
\begin{lemma}
\label{lm:GeneralRouting2}
Let $R$ be a routing graph on
 $\GridLabel:=\Grid{\ell_i}_{i\in [d]}$ with $d\geq 2$ and aspect ratio~$\gM$.
If $\bpi_h(R)$ is a $p$-$q$ routing graph
 with node set $\bpi_h(V(\GridLabel))$
 for some $h\in [d]$, then
 $R$ can be routed on $\GridLabel$
 with a dilation at most $2\sum_{i=1}^d \ell_i$ and
 an edge-congestion at most
 $2\lceil\gM\cdot\max\{p,q\}\rceil$.
\end{lemma}

\subsection{Separator-Based Embedding}
\label{ssc:SeparatorBasedEmbedding}
The following is our core theorem:
\begin{theorem}
\label{th:EmbedCore}
Suppose that $G$ is a graph with
 $N$ nodes, maximum node degree $\Delta$, and with a $\beta$-decomposition tree
 of expansion $Cn^\alpha$ ($C>0$, $0\leq\alpha< 1$, $1/2\leq\beta<1$),
 and that
$\GridLabel$ is a grid
 with a dimension $d\geq 2$,
 at least $N$ nodes,
 and with constant aspect ratio.
Then, $G$ can be embedded into $\GridLabel$
 with a dilation of $O(dN^{1/d})$, and with an edge-congestion of
 $O(dC+d^2\Delta)$ if $d>2/(1-\alpha)$,
 $O(C/(1-\alpha-\frac{1}{d})+d^2\Delta)$ if $1/(1-\alpha)<d\leq 2/(1-\alpha)$,
 and $O(C(N^{\alpha-1+\frac{1}{d}}+\log N)+d^2\Delta)$ if $d\leq 1/(1-\alpha)$,
\end{theorem}
In fact, we can obtain
 Theorem~\ref{th:EmbedGeneral} by combining
 Theorem~\ref{th:EmbedCore} with Lemma~\ref{lm:DecompositionTreeExpansion}.
If $G$ is a graph with $N$ nodes, maximum node degree $\Delta$,
 and with a $\beta$-node-separator of size $O(n^\alpha)$, then
 by Lemma~\ref{lm:DecompositionTreeExpansion},
 $G$ has a $\frac\beta{1-\epsilon}$-decomposition tree with expansion
 $O(\Delta n^\alpha/\epsilon)=O(\Delta n^\alpha)$
 for any $0<\epsilon<1-\beta$.
By Theorem~\ref{th:EmbedCore}, therefore,
 $G$ can be embedded into $\GridLabel$ with a dilation of $O(dN^{\frac 1 d})$,
 and with an edge-congestion of
 $O(\Delta\cdot\max\{d,1/(1-\alpha-\frac{1}{d})\}+d^2\Delta)=O(\Delta)$
 if 
 $d>1/(1-\alpha)$ is fixed, 
 $O(\Delta(N^{\alpha-1+\frac{1}{d}}+\log N)+d^2\Delta) =O(\Delta\log N)$
 if 
 $d=1/(1-\alpha)$, 
 and
 $O(\Delta(N^{\alpha-1+\frac{1}{d}}+\log N)+d^2\Delta)
 =O(\Delta N^{\alpha-1+\frac{1}{d}})$
 if 
 $d<1/(1-\alpha)$.

We prove Theorem~\ref{th:EmbedCore} by constructing
 a desired embedding algorithm, called SBE.
We first outline ideas and analysis of SBE,
 then specify the definition of SBE,
 and finally prove the correctness and the edge-congestion.

\subsection*{Proof Sketch}
We describe
 a proof sketch
 for the case $1/(1-\alpha)<d\leq 2/(1-\alpha)$ since
 essential part of idea appears in this case.
Basically, we partition $G$ according to its decomposition tree,
 recursively embed the partitioned subgraphs of $G$ into 
 partitioned subgrids of the host grid
 $\GridLabel:=\Grid{\ell_i}_{i\in [d]}$,
 and route cut edges, i.e., edges removed in partitioning $G$.
In order to avoid an edge of $\GridLabel$ being used in too many recursive
 steps,
 we route the cut edges on one of edge-disjoint subgraphs of $\GridLabel$,
 called \emph{channels}.
The channel associated with a positive integer $w$ roughly equal to
 $\frac{1}{2}(1-\alpha-\frac{1}{d})\log_2 N$
 is a grid-like graph homeomorphic to 
 $\GridLabel':=\Grid{\frac{\ell_1}{2^w},\frac{\ell_2}{2^w},\ell_3,\ldots,\ell_d}$
 and
 induced by
 the nodes $v\in V(\GridLabel)$
 with $\pi_i(v)\equiv 2^{w-1} \pmod{2^w}$ for each $i=1,2$.
We can find the channel
 in $\GridLabel$ as a non-empty subgraph
 if $d\leq 2/(1-\alpha)$.
When we embed an $n$-node subgraph $H$ of $G$ appeared
 in the decomposition tree,
 we partially route the external edges of each child of $H$
 to the channel
 associated with $w\simeq\frac{1}{2}(1-\alpha-\frac{1}{d})\log_2 n$
 and route the cut edges of $H$ by
 connecting the two sets of the external edges of children of $H$
 on this channel.
We here say ``partially'' in two meanings:
One meaning is that external edges are viewed as
 half-edges just leaving a child of $H$
 and are routed halfway.
The other is that 
 an external edge leaving a node in the decomposition tree
 is also an external edge of some descendants
 and is routed step by step among recursive steps.
I.e.,
 a cut edge is routed by connecting two partially routed
 external edges of the children,
 which are recursively routed using partially routed external edges
 of grandchildren, and so on.
Consequently, 
 cut edges of $H$ are routed through channels associated with
 integers up to $\frac{1}{2}(1-\alpha-\frac{1}{d})\log_2 n$
 in recursive steps from base embeddings to the embedding of~$H$.

The section of $\GridLabel'$ across a dimension $i$ is
 a $(d-1)$-dimensional grid with node set $\bpi_i(V(\GridLabel'))$.
If $\GridLabel'$ is associated with
 $w=\frac{1}{2}(1-\alpha-\frac{1}{d})\log_2 n$, then
 the minimum size $S$ of the section is
 $\min_i|\bpi_i(V(\GridLabel'))|
 =\Omega(n^{(d-1)/d}/2^{2w})=\Omega(n^\alpha)$.
Since $H$ and children of $H$ have at most $Cn^\alpha$ external edges,
 we route the external edges (or half-edges) of each child
 to the channel so that at most $D:=Cn^\alpha/S=O(C)$ halfway points $v$ 
 have the same $\bpi_i(v)$, where
 $i$ minimizes $|\bpi_i(V(\GridLabel'))|$.
We can inductively observe that this routing can be done
 with an edge-congestion $O(D)$
 using Lemma~\ref{lm:GeneralRouting}.
The integer $w$ decreases after $P=O(1/(1-\alpha-\frac{1}{d}))$
 recursive steps proceed because
 a guest graph becomes roughly half in one recursive step.
Therefore, the total edge-congestion incurred by entire recursive steps is
 at most $O(PD)=O(C/(1-\alpha-\frac{1}{d}))$
 since channels associated with different $w$'s are edge-disjoint.
We stop the recursive procedure at the point $\min\{\ell_i\}=\Theta(d)$,
 by which
 we can obtain a base embedding with an edge-congestion $B=O(d\Delta)$ using
 Lemmas \ref{lm:PermutationToEmbedding}--\ref{lm:GeneralRouting2}.
Because an edge of $\GridLabel$ can possibly be contained
 in $O(d)$ base embeddings, the total edge-congestion is
 at most $O(PD+dB)=O(C/(1-\alpha-\frac{1}{d})+d^2\Delta)$.

The reason of the limit $\min\{\ell_i\}=\Theta(d)$ of recursive procedure
 is as follows:
We cannot always partition the host grid with
 a ``flat'' section due to the difference between
 the number of nodes of a partitioned guest graph and multiples of
 the size of the section.
In our algorithm, therefore,
 we partition a host grid into two subgrids that
 may share a $(d-1)$-dimensional grid as ``ragged'' sections.
Such a $(d-1)$-dimensional grid might be used as channels in two partitioned
 grids
 during $O(d)$ recursive steps in the worst case,
 which would yield a $2^{O(d)}$ factor in the edge-congestion.
To avoid this, we actually remove any boundary 
 of a host grid from a channel,
 so that two partitioned grids have disjoint channels.
However, we might have an exponential factor again
 if we would continue the recursive procedure
 until $\min\{\ell_i\}$ is much smaller than~$d$.
For instance, if $\min\{\ell_i\}=O(1)$, then
 removal of the boundary 
 for each dimension would shrink
 the channel exponentially, implying
 $S=n^\alpha/2^{O(d)}$ and hence $D=O(C2^{O(d)})$.

\subsection*{Definition of SBE}
Suppose that $G_0$
 and $\GridLabel_0$ are a guest graph and a host grid, respectively, satisfying
 the conditions of Theorem~\ref{th:EmbedCore}.
Let $\mathcal T$ be a $\beta$-decomposition tree for $G_0$ with
 expansion~$Cn^\alpha$.
We set $\gM$ to the larger number of
 the aspect ratio of $\GridLabel_0$ and
\[
\frac{1}{1-\beta}\left(\frac{1}{7\beta}+e\beta\right)+\frac{5}{4}>4,
\]
 where
 $e$ is base of the natural logarithm.
We assume that
 any proper subgrid of $\GridLabel_0$ has less than $N$ nodes.

\begin{figure}
\begin{center}
\includegraphics[scale=1.0]{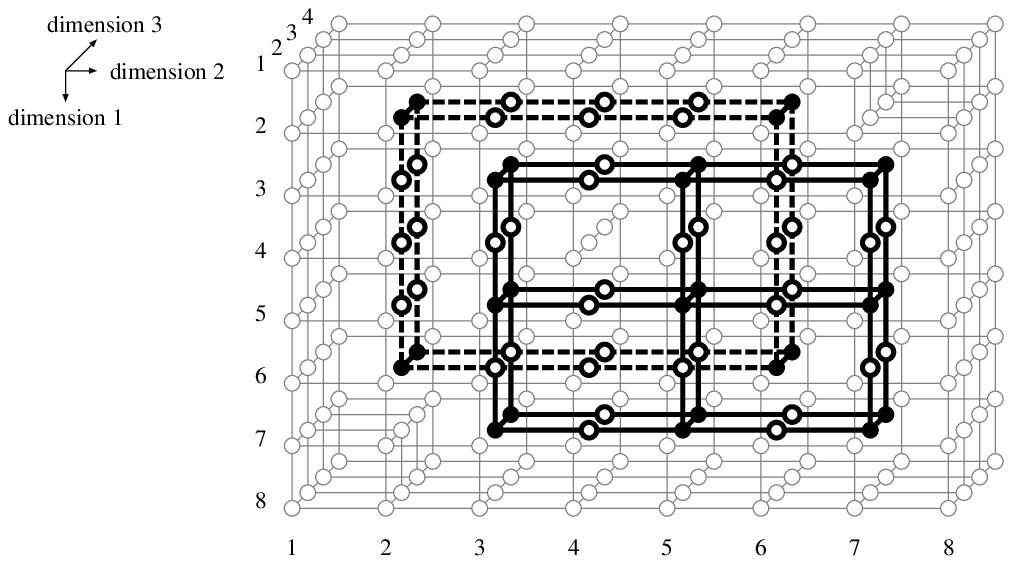}
\end{center}
\caption{%
Channels of $\CF \GridLabel 1$ and $\CF \GridLabel 2$ for
 $\GridLabel=\Grid{8,8,4}$.
Black nodes are contained in $\CF \GridLabel 1$ or $\CF \GridLabel 2$.
Dashed lines represent dimension-$1$ and -$2$ edges of the channel of 
 $\CF \GridLabel 2$.
}
\label{fig:channel}
\end{figure}

We use the following notations to define SBE formally:
Let $\V{w}:=\{v\in V(\GridLabel_0)\mid\pi_i(v)\equiv 2^{w-1}\pmod{2^w}
 \text{ for each $i=1,2$}\}$ for an integer $w\geq 1$, and let
 $\V 0:=V(\GridLabel_0)$.
For a $d$-dimensional subgrid $\GridLabel$ of $\GridLabel_0$,
 let
 $\CF \GridLabel w:=\{v\in \V{w}\cap V(\GridLabel)\mid
 \deg_\GridLabel(v)=2d\}$, where
 $\deg_\GridLabel(\cdot)$ is the node degree in $\GridLabel$.
The \emph{channel of $\CF \GridLabel w$}
 is the subgraph of $\GridLabel$ homeomorphic to
 a $d$-dimensional grid having $\CF \GridLabel w$ as grid points.
Specifically, this graph is induced by the node set
 $\CF\GridLabel w\cup
 \{s\in V(\GridLabel)\mid\exists i\in\{1,2\}\
 \exists\{u,v\}\subseteq\CF\GridLabel w\
 \pi_i(u)<\pi_i(s)<\pi_i(v)=\pi_i(u)+2^w,\ \bpi_i(u)=\bpi_i(s)=\bpi_i(v)\}$.
Two channels in $\Grid{8,8,4}$ are illustrated in Fig.~\ref{fig:channel}.
It should be noted that for any $w>w'\geq 1$,
 $\CF\GridLabel w\cap\CF\GridLabel {w'}=\emptyset$
 and channels of $\CF\GridLabel w$ and $\CF\GridLabel {w'}$
 are edge-disjoint.
The \emph{direction of $\CF{\GridLabel}{w}$} is a dimension
 $i\in [d]$ minimizing $\CS \GridLabel i w:=|\bpi_i(\CF \GridLabel w)|$.
In other words,
 the direction is a dimension of the longest side length
 of a grid having grid points $\CF{\GridLabel}{w}$.
In Fig.~\ref{fig:channel},
 the channel for $w=1$ has direction
 $1$ or $2$ because
 $\CS \GridLabel 1 1=\CS \GridLabel 2 1=6$ and $\CS \GridLabel 3 1=9$.
A mapping $\psi: X\rightarrow \CF \GridLabel w$ is said to be
 \emph{uniform across dimension $i$} if
 $\psi(X)$ are uniformly distributed on $\bpi_i(\CF \GridLabel w)$, i.e.,
 $\lambda_i(\psi):=\max_{v\in\bpi_i(\CF\GridLabel w)}
 |\{s\in X\mid \bpi_i(\psi(s))=v\}|=\lceil|X|/\CS \GridLabel i w\rceil$.
In Fig.~\ref{fig:channel}, for example,
 if $\psi:[4]\rightarrow\CF \GridLabel 1$ maps
 $1,2,3,4$ to $(3,3,2),(3,5,2),(5,3,3),(5,5,3)$, respectively, then
 $\psi$ is uniform across dimensions both $1$ and $3$ but not dimension $2$
 because
 $\lambda_1(\psi)=1=\lceil|[4]|/\CS\GridLabel 1 1\rceil=\lceil 4/6\rceil$,
 $\lambda_2(\psi)=2>\lceil|[4]|/\CS\GridLabel 2 1\rceil=\lceil 4/6\rceil$,
 and
 $\lambda_3(\psi)=1=\lceil|[4]|/\CS\GridLabel 3 1\rceil=\lceil 4/9\rceil$.
We note here that for any two dimensions $i$ and $j$,
 we can construct $\psi$ uniform across dimensions both $i$ and $j$ by
 uniformly distributing $\psi(X)$ among nodes
 on a $(d-1)$-dimensional diagonal hyperplane between dimensions $i$ and $j$
 in $\CF{\GridLabel} w$.
 
\newcounter{SBEStep}
\newcounter{LabelSBEInputOutput}
\newcounter{LabelSBEChannelConfiguration}
\newcounter{LabelSBEBaseEmbedding}
\newcounter{LabelSBESeparation}
\newcounter{LabelSBEResizingChannel}
\newcounter{LabelSBERecursiveEmbedding}
\newcounter{LabelSBERoutingEdges}
\subsubsection*{Step~\theSBEStep---Input and Output}
\setcounter{LabelSBEInputOutput}{\value{SBEStep}}
The formal input and output of SBE is as follows:
\paragraph*{Algorithm {\ttfamily SBE}($G$, $X$, $\GridLabel$, $U$)}
\begin{description}
\item[Input]\mbox{}
\begin{itemize}
\item
An $n$-node subgraph $G$ of $G_0$ contained in $V(\mathcal T)$.
\item
A multiset $X$ of nodes of $G$ incident to distinct
 external edges of $G$, i.e.,
 a node appears in $X$ as many times as the number of
 the external edges incident to the node.
\item
A subgrid $\GridLabel=\Grid{\ell_i}_{i\in [d]}$ of $\GridLabel_0$ with
 aspect ratio at most $\gM$, together with
 a set $U\subseteq V(\GridLabel)$ such that
 $U\supseteq\{v\in V(\GridLabel)\mid\deg_\GridLabel(v)=2d\}$ and
 $|U|=n$.
Suppose that $h$ is a dimension such that $\ell_h=\max_{i\in [d]}\{\ell_i\}$.
\end{itemize}
\item[Output]\mbox{}
\begin{itemize}
\item
An embedding $\langle\phi,\rho\rangle$ 
 of $G$ into $\GridLabel$
such that $\phi(V(G))=U$.
\item
A mapping $\psi:X\rightarrow \CF \GridLabel w$ uniform
 across the direction $k$ of $\CF \GridLabel w$, where
 $w\geq 0$ is an integer defined in Step~1.
\item
A routing 
 $\sigma$ of the routing graph with
 node set $\phi(X)\cup\psi(X)$ and
 edge set $\{(\phi(u),\psi(u))\mid u\in X\}$.
\end{itemize}
\end{description}

Initially, we arbitrarily choose $U$ as desired and
 perform
 {\ttfamily SBE}($G_0$, $\emptyset$, $\GridLabel_0$, $U$).

\stepcounter{SBEStep}
\subsubsection*{Step~\theSBEStep---Channel Configuration}
\setcounter{LabelSBEChannelConfiguration}{\value{SBEStep}}
This step sets an integer $w$, by which we configure the channel
 of $\CF\GridLabel w$ to route $\rho$ and~$\sigma$.
We define $w:=\max\{
\lfloor\frac{1}{2}
((1-\tilde\alpha-\frac{1}{d})\log_2 n-\log_2\frac\gM{1-\beta})\rfloor,0\}$,
 where
 $\tilde\alpha:=\max\{1-\frac 2 d,\alpha\}$.
In other words,
 $w=\max\{
\lfloor\frac{1}{2}
(\frac 1 d \log_2 n-\log_2\frac\gM{1-\beta})\rfloor,0\}$
 if $d>2/(1-\alpha)$,
 $w=\max\{\lfloor\frac{1}{2}
((1-\alpha-\frac{1}{d})\log_2 n-\log_2\frac\gM{1-\beta})\rfloor,0\}$
 if $1/(1-\alpha)<d\leq 2/(1-\alpha)$, and
 $w=0$ if $d\leq 1/(1-\alpha)$.

\stepcounter{SBEStep}
\subsubsection*{Step~\theSBEStep---Base Embedding}
\setcounter{LabelSBEBaseEmbedding}{\value{SBEStep}}
If 
 $\ell_h\leq 2\gM d$, then
 SBE does not call itself recursively any longer
 and constructs a base embedding as follows:
\begin{enumerate}
\item
If $X=\emptyset$, then
 let $\phi:V(G)\rightarrow U$ be an arbitrary one-to-one mapping.
Otherwise,
 let $Y$ be the set (not a multiset) of nodes incident to external edges of
 $G$.
We construct a one-to-one mapping $\phi:V(G)\rightarrow U$ so that
 degrees of nodes in $Y$ are uniformly distributed on $\bpi_k(U)$.
Specifically,
 for any $v\in\bpi_k(U)$,
 the sum of degrees of nodes $s\in Y$ with $\bpi_k(\phi(s))=v$
 is at most $e|X|/|\bpi_k(U)|+\Delta$.
We prove later in Lemma~\ref{lm:BaseEmbedding} that
 $\phi$ can be constructed as desired.
This construction implies that
 if $R$ is the routing graph to be routed by $\sigma$, then
 $\bpi_k(R)$ has maximum outdegree at most
 $e|X|/|\bpi_k(U)|+\Delta\leq e|X|/\CS{\GridLabel} k w+\Delta$.
\item
Construct a mapping $\psi:X\rightarrow \CF \GridLabel w$ so that
 $\psi$ is uniform across dimension $k$, implying
 that
 $\bpi_k(R)$ has maximum indegree at most
 $\lceil |X|/\CS{\GridLabel} k w\rceil$.
\item
Apply Lemmas \ref{lm:PermutationToEmbedding} and~\ref{lm:GeneralRouting}
 on $\GridLabel$ to obtain $\rho$.
\item
If $X\neq\emptyset$, then
 apply Lemma~\ref{lm:GeneralRouting2} on $\GridLabel$
 to obtain $\sigma$.
\item
Return.
\end{enumerate}

\begin{figure}
\begin{center}
\includegraphics[scale=1.0]{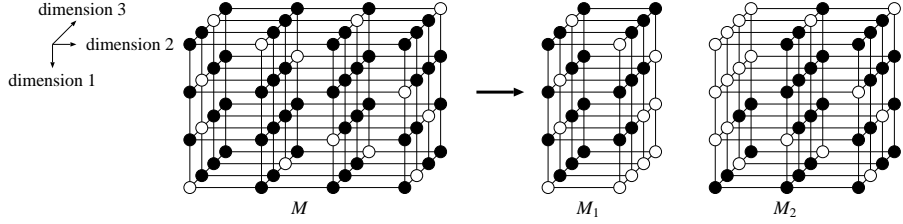}
\end{center}
\caption{An example of partitioning a host grid $\GridLabel$
 with 50 nodes in $U$ (represented by black nodes)
 into
 a grid $\GridLabel_1$ with $|U_1|=20$
 and
 a grid $\GridLabel_2$ with $|U_2|=30$.
}
\label{fig:GridPartition}
\end{figure}
 
\stepcounter{SBEStep}
\subsubsection*{Step~\theSBEStep---Partition}
\setcounter{LabelSBESeparation}{\value{SBEStep}}
Suppose $\ell_h> 2\gM d$.
Let $G_1$ and $G_2$ be children of $G$ in $\mathcal T$ and have
 $n_1$ and $n_2$ nodes, respectively.
Now $\GridLabel$ is partitioned
 into two subgrids
 $\GridLabel_1:=\Grid{\ell_1,\ldots,\ell_{h-1},m_1,\ell_{h+1},\linebreak[0]\ldots,\ell_d}$ and
 $\GridLabel_2:=\Grid{\ell_1,\ldots,\ell_{h-1},m_2,\ell_{h+1},\ldots,\ell_d}$, together
 with node sets
 $U_1\subseteq V(\GridLabel_1)$ and $U_2\subseteq V(\GridLabel_2)$
 such that
 $m_1+m_2=\ell_h+1$,
 $U_1\cup U_2=U$,
 $U_1\cap U_2=\emptyset$,
 and
 $|U_j|=n_j$ for $j=1,2$ (Fig~\ref{fig:GridPartition}).
We here duplicate the $(d-1)$-dimensional grid induced by
 $\{v\in V(\GridLabel)\mid\pi_h(v)=m_1\}$ to be shared
 by $\GridLabel_1$ and $\GridLabel_2$,
 so that $m_1+m_2$ equals not $\ell_h$ but $\ell_h+1$.
We do this because
 $\GridLabel_1$ and $\GridLabel_2$ must have
 enough numbers of nodes in $U$
 onto which $V(G_1)$ and $V(G_2)$ can be mapped, respectively.
In Fig~\ref{fig:GridPartition}, actually,
 however we partition $\GridLabel$ so that $m_1+m_2=l_h$,
 either $|V(\GridLabel_1)\cap U|<20$ or
 $|V(\GridLabel_2)\cap U|<30$.
We will prove later in Lemma~\ref{lm:GridPartition} that
 the resulting subgrids have aspect ratio at most $\mu$.

\stepcounter{SBEStep}
\subsubsection*{Step~\theSBEStep---Recursive Embedding}
\setcounter{LabelSBERecursiveEmbedding}{\value{SBEStep}}
This step recursively
 embeds $G_1$ and $G_2$ into
 $\GridLabel_1$ and $\GridLabel_2$, respectively.
We also construct a routing $\tilde\sigma_j$ for $j=1,2$,
 which draws the external edges of $G_j$ to $\tilde\psi_j(X_j)$.
Here,
 $X_j$ is the multiset of nodes of $G_j$ incident to
 distinct external edges of $G_j$, and
 $\tilde\psi_j: X_j\rightarrow \CF{\GridLabel_j}{w}$
 is a mapping uniform across dimension $k$, i.e.,
 the direction of $\CF \GridLabel{w}$.
With this routing, in the subsequent step,
 we will make a routing graph
 with sources and targets in $\tilde\psi_1(X_1)$ and
 $\tilde\psi_2(X_2)$, respectively,
 on the channel of $\CF \GridLabel{w}$ to connect cut edges.
Projecting this routing graph along dimension $k$ yields
 a $\lceil|X_1|/\CS\GridLabel k w\rceil$-$\lceil|X_2|/\CS\GridLabel k w\rceil$
 routing graph.
Thus, the cut edges will be routed on the channel of
 $\CF{\GridLabel} w$
 with an edge-congestion of
 $2\cdot\max_{j=1,2}\lceil|X_j|/\CS\GridLabel k w\rceil$
 using Lemma~\ref{lm:GeneralRouting}.
In applying Lemma~\ref{lm:GeneralRouting},
 we regard the underlying channel as the homeomorphic grid with
 grid points~$\CF\GridLabel w$.
A detailed analysis for this edge-congestion will later be provided in
 Lemmas~\ref{lm:PartialRouting}-\ref{lm:D}.
We aim to suppress the edge-congestion of $\tilde\sigma_j$
 in a similar way.
Therefore,
 we make $\tilde\psi_j$ uniform across
 not only dimension $k$ but also
 the direction $k_j$ of $\CF{\GridLabel_j}{w_j}$, where
 $w_j$ is $w$ computed for $n_j$ in the recursive procedure,
 since
 $\psi_j$ is made uniform across dimension $k_j$
 in the recursive procedure (Fig~\ref{fig:RecursiveEmbedding}).
We need to treat two more matters in constructing $\tilde\sigma_j$.

\begin{figure}
\begin{center}
\includegraphics[scale=1.0]{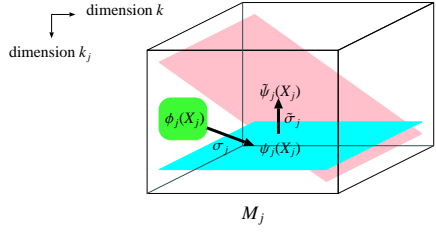}
\end{center}
\caption{Mapping $\tilde\psi_j$ and routing $\tilde\sigma_j$ 
 constructed in Step~{\theLabelSBERecursiveEmbedding}.
Symbols $\phi_j$, $\sigma_j$, and $\psi_j$ are the output
 $\phi$, $\sigma$, and $\psi$ in the embedding of $\GridLabel_j$, respectively.
The horizontal and diagonal planes represent the uniformness of
 $\psi_j(X_j)$ and $\tilde\psi_j(X_j)$, respectively.}
\label{fig:RecursiveEmbedding}
\end{figure}

First, we need a channel (a grid-like graph) containing both
 $\psi_j(X_j)\subseteq\CF{\GridLabel_j}{w_j}$ and
 $\tilde\psi_j(X_j)\subseteq\CF{\GridLabel_j}{w}$.
Just taking union of two channels of
 $\CF{\GridLabel_j}{w_j}$ and $\CF{\GridLabel_j}{w}$ would not suffice
 because $\CF{\GridLabel_j}{w_j}$ and $\CF{\GridLabel_j}{w}$ are disjoint
 if $w_j\neq w$.
We define
 the channel for $w_j$ and $w$
 as the graph homeomorphic to a $d$-dimensional grid having
 $\CF{\GridLabel_j}{w_j,w}:=\{v\in V(\GridLabel_0)\mid
 \pi_i(v)\equiv 2^{w_j-1}\ \text{or}\  2^{w-1}\pmod{2^w}\ \text{for $i=1,2$}\}
 \cap\{v\in V(\GridLabel_j)\mid\deg_{\GridLabel_j}(v)=2d\}$
 as grid points.
It should be noted that
 if $w>w_j>0$, then
 edges contained in the channel of 
 $\CF{\GridLabel_j}{w_j,w}$ but in
 the channel of neither $\CF{\GridLabel_j}{w_j}$ nor $\CF{\GridLabel_j}{w}$
 are uniquely determined by $w_j$ and $w$ and not contained any other
 channel in $\GridLabel_j$
 except the channel of $\CF{\GridLabel_j}0$.

The second matter is that
 the direction of $\CF{\GridLabel_j}{w_j,w}$ 
 may differ from~$k_j$,
 across which $\psi_j$ and $\tilde\psi_j$ are uniform.
This means that just applying Lemma~\ref{lm:GeneralRouting}
 would not guarantee a desired edge-congestion.
If we applied the algorithm of Lemma~\ref{lm:GeneralRouting}
 on the channel of $\CF{\GridLabel_j}{w_j,w}$, then
 the algorithm would be recursively called in the non-increasing order of
 side lengths of a grid having $\CF{\GridLabel_j}{w_j,w}$
 as grid points.
We here modify the order by replacing $\CF{\GridLabel_j}{w_j,w}$ with
 $\CF{\GridLabel_j}{w_j}$.
The modified algorithm yields a desired edge-congestion
 as we will later prove in Lemma~\ref{lm:PartialRouting}.

For each $j=1,2$, specifically, SBE performs the following procedures:

\begin{enumerate}
\item
Call {\ttfamily SBE}($G_j$, $X_j$, $\GridLabel_j$, $U_j$).
Let $\phi_j$, $\rho_j$, $\psi_j$, and $\sigma_j$ denote
 the output $\phi$, $\rho$, $\psi$, and $\sigma$
 of the recursive call, respectively.
\item
Construct a mapping $\tilde\psi_j:X_j\rightarrow\CF{\GridLabel_j}{w}$
 uniform across dimensions both $k$ and~$k_j$,
 where $k_j$ is the direction of $\CF{\GridLabel_j}{w_j}$, and
 $w_j:=\max\{\lfloor\frac{1}{2}
((1-\tilde\alpha-\frac{1}{d})\log_2 n_j-\log_2\frac\gM{1-\beta})\rfloor,0\}$.
\item
Let $\tilde\sigma_j$ be a routing from $\psi_j(X_j)$ to $\tilde\psi_j(X_j)$
 on the channel of $\CF{\GridLabel_j}{w_j,w}$
 obtained by using the modified algorithm of
 Lemma~\ref{lm:GeneralRouting} in which
 we recursively call the algorithm in the non-increasing order
 of side lengths of a grid having $\CF{\GridLabel_j}{w_j}$ as grid points.
\end{enumerate}

\stepcounter{SBEStep}
\subsubsection*{Step~\theSBEStep---Routing Cut and External Edges}
\setcounter{LabelSBERoutingEdges}{\value{SBEStep}}
This step constructs $\psi$, then
 completes $\rho$ and $\sigma$ using
 $\tilde\psi_j$ and $\psi$.
The routings $\rho$ and $\sigma$ are obtained
 simply using Lemma~\ref{lm:GeneralRouting}
 on the channel of $\CF\GridLabel w$ since
 $\tilde\psi_j(X_j)\subseteq\CF{\GridLabel_j}{w}\subseteq\CF\GridLabel{w}$.
The following are specific procedures of this step:
\begin{enumerate}
\item
Construct a mapping $\psi:X\rightarrow \CF \GridLabel w$
 uniform across dimension $k$.
\item
By using Lemma~\ref{lm:GeneralRouting}, construct
 $\tilde\sigma$ for the routing graph on
 the channel of $\CF \GridLabel{w}$ with node set
 $\tilde\psi_1(X_1)\cup\tilde\psi_2(X_2)$
 and edge set $\{(\tilde\psi_1(s_1),\tilde\psi_2(s_2))\mid
 s_1\in X_1\setminus X,\ s_2\in X_2\setminus X,\ (s_1,s_2)\in E(G)\}\cup
 \bigcup_{j=1,2}\{(\tilde\psi_j(s),\psi(s))\mid s\in X_j\cap X\}$
 (Fig~\ref{fig:RoutingEdges}).
It should be noted that
 $(s_1,s_2)\in E(G)$ with
 $s_1\in X_1\setminus X$ and $s_2\in X_2\setminus X$ is
 a cut edge of $G$, and
 $s\in X_j\cap X$ is a node incident to an external edge of $G$.

\begin{figure}
\begin{center}
\includegraphics[scale=1.0]{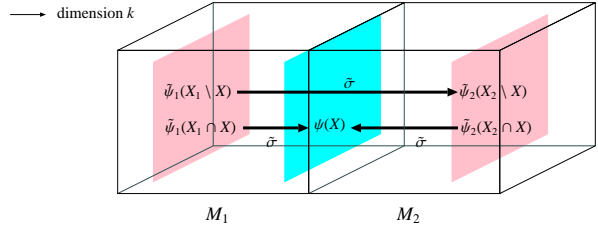}
\end{center}
\caption{Mapping $\psi$ and routing $\tilde\sigma$ constructed in
 Step~{\theLabelSBERoutingEdges}.
The left, center, and right planes represent the uniformness of
 $\tilde\psi_1(X_1)$,
 $\psi(X)$, and
 $\tilde\psi_2(X_2)$, respectively.}
\label{fig:RoutingEdges}
\end{figure}

\item
Let $\rho$ map the cut edges of $G$ onto paths obtained by
 concatenating the images of
 $\sigma_1$, $\tilde\sigma_1$, $\tilde\sigma$, $\tilde\sigma_2$, and $\sigma_2$.
Specifically, for $(s_1,s_2)\in E(G)$ with
 $s_1\in X_1\setminus X$ and $s_2\in X_2\setminus X$,
 let
\[
\begin{split}
\rho(\phi(s_1),\phi(s_2))
 := &\bigcup_{j=1,2}\left(\sigma_j(\phi(s_j),\psi_j(s_j))
 \cup\tilde\sigma_j(\psi_j(s_j),\tilde\psi_j(s_j))\right)\\
 &\quad \cup\tilde\sigma(\tilde\psi_1(s_1),\tilde\psi_2(s_2)).
\end{split}
\]
\item
Let $\sigma$ map the external edges of $G$ onto paths obtained by
 concatenating the images of
 $\sigma_j$, $\tilde\sigma_j$, $\tilde\sigma$.
Specifically,
 for $s\in X_j\cap X$ ($j=1,2$),
 let 
\[
\begin{split}
\sigma(\phi(s),\psi(s)):= & \sigma_j(\phi(s),\psi_j(s))
 \cup\tilde\sigma_j(\psi_j(s),\tilde\psi_j(s))
 \cup\tilde\sigma(\tilde\psi_j(s),\psi(s)).
\end{split}
\]
\end{enumerate}

\subsection*{Correctness}
To prove that SBE yields an output satisfying
 the conditions specified in Step~{\theLabelSBEInputOutput},
 we first prove in Lemma~\ref{lm:BaseEmbedding} below that
 we can construct $\phi$ as desired in Step~{\theLabelSBEBaseEmbedding}
 for the case $X\neq\emptyset$.
We then prove
 in Lemmas \ref{lm:GridPartition} and~\ref{lm:ChannelNonEmpty} below
 that $\GridLabel_j$ 
 defined in Step~{\theLabelSBESeparation}
 has aspect ratio at most~$\gM$, and that
 $\CF{\GridLabel_j} w$ is non-empty as well as $\CF\GridLabel w$.
These facts guarantee that a valid input is given to a child procedure
 in Step~{\theLabelSBERecursiveEmbedding}, and
 that mappings $\tilde\psi_j$ and $\psi$ can be constructed
 in Steps {\theLabelSBERecursiveEmbedding} and~{\theLabelSBERoutingEdges},
 respectively.
\begin{lemma}
\label{lm:BaseEmbedding}
Suppose $X\neq\emptyset$ in Step~{\theLabelSBEBaseEmbedding}.
We can construct a one-to-one mapping $\phi:V(G)\rightarrow U$ such that
 for any $v\in\bpi_k(U)$,
 the sum of degrees of nodes $s\in Y$ with $\bpi_k(\phi(s))=v$
 is at most $e|X|/|\bpi_k(U)|+\Delta$, where
 $Y$ is the set (not a multiset) of nodes incident to external edges of $G$.
\end{lemma}
\begin{proof}
If $U=V(\GridLabel)$, then
 we can map $Y$ onto $U$ in a trivial manner so that
 $\max_{v\in\bpi_k(U)}\linebreak[0]
\sum_{s\in Y,\ \bpi_k(\phi(s))=v}\deg_{G}(s)\leq
 \sum_{s\in Y}\deg_{G}(s)/|\bpi_k(U)|+\Delta=|X|/|\bpi_k(U)|+\Delta$.
This is because this mapping can be viewed as
 a packing of $|Y|$ items of size at most $\Delta$
 to $|\bpi_k(U)|$ bins that
 can contain the same number $|U|/|\bpi_k(U)|=\ell_k$ of items.
If $U\subset V(\GridLabel)$, i.e.,
 $U$ does not contain some nodes on the boundary of $\GridLabel$, then
 some of the bins cannot contain $\ell_k$ items.
An upper bound can be obtained in the assumption that
 $U$ contains no node on the boundary of $\GridLabel$, and that
 we must map $Y$
 onto $\prod_{i\in [d]\setminus\{k\}}(\ell_i-2)$ bins, in which
 the mapping is not one-to-one if $|Y|>\prod_{i\in [d]}(\ell_i-2)$.
Thus, we have
\begin{equation}
\label{eq:DegreeDistribution}
\begin{split}
\max_{v\in\bpi_k(U)}\sum_{\substack{s\in Y\\\bpi_k(\phi(s))=v}}\deg_{G}(s)
&\leq\frac{\sum_{s\in Y}\deg_{G}(s)}{\prod_{i\in [d]\setminus\{k\}}(\ell_i-2)}
 +\Delta
\leq\frac{|X|}{\prod_{i\in [d]\setminus\{k\}}(\ell_i-2)}\cdot
\frac{\prod_{i\in [d]\setminus\{k\}}\ell_i}{|\bpi_k(U)|}+\Delta.\\
\end{split}
\end{equation}
We have $X\neq\emptyset$ only if
 the current base embedding is called by a parent procedure.
This implies
 $\min_{i\in [d]}\ell_i>2d$ as proved in Lemma~\ref{lm:GridPartition} below,
 and therefore,
\[
\frac{\prod_{i\in [d]\setminus\{k\}}\ell_i}{\prod_{i\in [d]\setminus\{k\}}(\ell_i-2)}
=\prod_{i\in [d]\setminus\{k\}}\frac{\ell_i}{\ell_i-2}
<\left(\frac{2d}{2d-2}\right)^{d-1}<e.
\]
Combined with (\ref{eq:DegreeDistribution}), we have the lemma.
\end{proof}
\begin{lemma}
\label{lm:GridPartition}
For $j=1,2$,
 $\GridLabel_j$ defined in Step~{\theLabelSBESeparation}
 has aspect ratio at most $\gM$
 and $\min\{\ell_1,\linebreak[0]\ldots,\ell_{h-1},m_j,\ell_{h+1},\ldots,\ell_d\}>2d$.
\end{lemma}
\begin{proof}
Assume without loss of generally that $m_1\leq m_2$.
Because $\GridLabel$ has aspect ratio at most $\gM$ and
 $\min_{i\in[d]}{\ell_i}>2d$
 (for otherwise,
 $\ell_h\leq\gM\cdot\min_{i\in[d]}{\ell_i}\leq 2\gM d$, and hence,
 SBE entered the base step),
 it suffices to prove that
 $m_1>2d$ and
 $\ell_h/m_1\leq\gM$.

Because $n_j=|U_j|\geq |\{v\in V(\GridLabel_j)\mid\deg_{\GridLabel_j}(v)=2d\}|$ for $j=1,2$,
 it follows that
 $(m_j-2)\prod_{i\in [d]\setminus\{h\}}(\ell_i-2)\leq n_j\leq
 m_j\prod_{i\in [d]\setminus\{h\}} \ell_i$, and hence,
\[
m_j-2\leq\frac{n_j}{\prod_{i\in [d]\setminus\{h\}}(\ell_i-2)}\leq
 m_j\prod_{i\in [d]\setminus\{h\}}\frac{\ell_i}{\ell_i-2}
 <m_j\left(\frac{2d}{2d-2}\right)^{d-1}<e m_j.
\]
We have by the inequalities that
\begin{equation}
\label{eq:GridPartition_mn}
\frac{m_2-2}{n_2}\leq\frac{1}{\prod_{i\in [d]\setminus\{h\}}(l_i-2)}
<\frac{e m_1}{n_1}.
\end{equation}

Because $n-n_1=n_2\leq\lceil\beta n\rceil\leq\beta n+1$, it follows that
\begin{equation}
\label{eq:GridPartition_n1n2}
n_1\geq(1-\beta)n-1\geq
(1-\beta)\frac{n_2-1}{\beta}-1=\frac{(1-\beta)n_2}{\beta}-\frac{1}{\beta}.
\end{equation}
Moreover, it follows that
$n\geq\prod_{i\in [d]}(\ell_i-2)>(2\gM d-2)(2d-2)^{d-1}
 \geq 8\gM-4$, which is larger than $7\gM$ because $\gM>4$.
Hence, it follows that
\begin{equation}
\label{eq:GridPartition_n1}
n_1\geq (1-\beta)n-1>7\gM(1-\beta)-1>7e\beta.
\end{equation}
Thus, by (\ref{eq:GridPartition_mn})--(\ref{eq:GridPartition_n1}),
\[
m_2-2<\frac{e m_1n_2}{n_1}
 \leq\frac{e m_1}{1-\beta}\left(\beta+\frac{1}{n_1}\right)
 <\frac{e m_1}{1-\beta}\left(\beta+\frac{1}{7e\beta}\right)
\leq\left(\gM-\frac{5}{4}\right)m_1,
\]
 by which we obtain
 $(\gM-\frac{1}{4})m_1> m_1+m_2-2=\ell_h-1>2\gM d-1$.
Therefore, 
\begin{align*}
m_1&>\frac{2\gM d-1}{\gM-\frac 1 4}=\frac{2d(\gM-\frac 1{2d})}{\gM-\frac 1 4}
\geq 2d,\ \text{and}\\
\frac{\ell_h}{m_1}&<\gM-\frac{1}{4}+\frac{1}{m_1}
<\gM-\frac{1}{4}+\frac{1}{2d}\leq\gM.
\end{align*}
\end{proof}
\begin{lemma}
\label{lm:ChannelNonEmpty}
For $j=1,2$, $\CF{\GridLabel_j} w$ in Step~{\theLabelSBERecursiveEmbedding}
 is non-empty.
\end{lemma}
\begin{proof}
Suppose
 $\GridLabel_j=\Grid{\ell^j_i}_{i\in [d]}$ and
 $\lmin^j:=\min_{i\in[d]}{\ell^j_i}$.
We can observe by the definition of $\CF{\GridLabel_j} w$
 that
 if $\lfloor(\ell^j_i-2)/2^w\rfloor>0$ for $i=1,2$, and if
 $\ell^j_i-2>0$ for $3\leq i\leq d$, then
 $\CF{\GridLabel_j} w$ is non-empty.
Because $\lmin^j>2d>4$ by Lemma~\ref{lm:GridPartition},
 the lemma holds if $w=0$.
Assume $w\geq 1$.
Then, the lemma is implied by $\lmin^j/2^w\geq 2$.
The assumption $w\geq 1$ implies that
 $d>1/(1-\alpha)$ by the definition of $w$, and that
\begin{equation}
\label{eq:ChannelNonEmpty_w}
2^w
 \leq\left(\frac{(1-\beta)n^{1-\tilde\alpha-\frac 1 d}}{\gM}\right)^{1/2}
 =\left(\frac{(1-\beta)n^{\min\{1/d,1-\alpha-\frac 1 d\}}}{\gM}\right)^{1/2}
 \leq\left(\frac{(1-\beta)n^{1/d}}{\gM}\right)^{1/2}.
\end{equation}
As estimated in (\ref{eq:GridPartition_n1n2}) and (\ref{eq:GridPartition_n1}),
 it follows that $n_j\geq (1-\beta)n-1$ and $n>7\gM$.
Therefore,
\begin{equation}
\label{eq:Ratio_njn}
n_j\geq\left(1-\beta-\frac 1 n\right)n>\left(1-\beta-\frac 1 {7\gM}\right)n
>\left(1-\beta-(1-\beta)\beta\right)n=(1-\beta)^2n.
\end{equation}
Because
 $\GridLabel_j$ has aspect ratio at most $\gM$
 by Lemma~\ref{lm:GridPartition},
 it follows that
\begin{equation}
\label{eq:nj_lmin}
 n_j^{1/d} \leq\max_{i\in[d]}\{l^j_i\}\leq\gM\lmin^j.
\end{equation}
Combining
 (\ref{eq:ChannelNonEmpty_w}), (\ref{eq:Ratio_njn}), and~(\ref{eq:nj_lmin}),
\[
\frac{\lmin^j}{2^w}
\geq\frac{n_j^{1/d}}{\gM}\left(\frac{\gM}{(1-\beta)n^{1/d}}\right)^{1/2}
>\left(\frac{(1-\beta)^{\frac 4 d -1}n^{1/d}}{\gM}\right)^{1/2}
\geq\left(\frac{(1-\beta)n^{1/d}}{\gM}\right)^{1/2}
\geq 2^w\geq 2.
\]
\end{proof}

\subsection*{Edge-Congestion}
We first estimate the edge-congestion of $\tilde\sigma_j$ 
 and $\tilde\sigma$
 in each recursive call of SBE.
Then, we prove the total edge-congestion.
In what follows, 
 for an $n$-node guest graph given to SBE as input,
 we will use $D^{w}(n)$
 to denote the maximum value of
 $\max_{i\in [d]}\lceil Cn^\alpha/\CS H i w\rceil$
 over all feasible $d$-dimensional host grids $H$.
\begin{lemma}
\label{lm:PartialRouting}
For $j=1,2$, $\tilde\sigma_j$ in Step~{\theLabelSBERecursiveEmbedding}
 imposes an edge-congestion at most
 $2D^w(n_j)$
 on the channel of $\CF{\GridLabel_j}{w_j,w}$.
\end{lemma}
\begin{proof}
As described in Step~{\theLabelSBERecursiveEmbedding},
 in constructing $\tilde\sigma_j$, we recursively
 call the algorithm of Lemma~\ref{lm:GeneralRouting} 
 in the non-increasing order of side lengths of a grid
 having not $\CF{\GridLabel_j}{w_j,w}$ but $\CF{\GridLabel_j}{w_j}$
 as grid points.
We can prove
 through a similar argument to that of Lemma~\ref{lm:GeneralRouting}
 that
 the modified algorithm achieves
 an edge-congestion of
 $2\lceil |X_j|/\CS{\GridLabel_j}{k_j}{w}\rceil
 \leq 2\lceil Cn_j^\alpha/\CS{\GridLabel_j}{k_j}{w}\rceil\leq 2D^w(n_j)$,
 noting that
 $\CF{\GridLabel_j}{w_j}\cup\CF{\GridLabel_j}{w}
\subseteq\CF{\GridLabel_j}{w_j,w}$
 and $|\pi_i(\CF{\GridLabel_j}{w_j})|\geq |\pi_i(\CF{\GridLabel_j}{w})|$
 for $i\in [d]$.
The following is an explicit proof.

Suppose that grids having
 $\CF{\GridLabel_j}{w_j}$ and $\CF{\GridLabel_j}{w_j,w}$
 as grid points are
 $\GridLabel':=\Grid{\ell'_1,\ldots,\ell'_d}$ and
 $\GridLabel'':=\Grid{\ell''_1,\ldots,\ell''_d}$, respectively.
I.e.,
 $\ell'_i=|\pi_i(\CF{\GridLabel_j}{w_j})|$ and
 $\ell''_i=|\pi_i(\CF{\GridLabel_j}{w_j,w})|$ for $i\in [d]$.
Because $\CF{\GridLabel_j}{w_j}\subseteq\CF{\GridLabel_j}{w_j,w}$,
 $\ell'_i\leq \ell''_i$ for each $i\in [d]$.
In this proof,
 we assume without loss of generality that
 $k_j=1$ and
 $\ell'_1=\ell'_{k_j}\geq\ell'_2\cdots\geq \ell'_{d}$.
Moreover, we regard the routing graph $R$
 from sources $\psi_j(X_j)$ to targets $\tilde\psi_j(X_j)$ on
 the channel of $\CF{\GridLabel_j}{w_j,w}$ as
 its corresponding routing graph on $\GridLabel''$.
Then, 
 $\bpi_{1}(R)$ is a $p$-$q$ routing graph
 with node set $\bpi_{1}(V(\GridLabel''))$,
 where $p:=\lceil|X_j|/\CS{\GridLabel_j}{k_j}{w_j}\rceil
=\lceil |X_j|/|\bpi_{k_j}(\CF{\GridLabel_j}{w_j})|\rceil$
 and $q:=\lceil|X_j|/\CS{\GridLabel_j}{k_j}{w}\rceil
=\lceil |X_j|/|\bpi_{k_j}(\CF{\GridLabel_j}{w})|\rceil$,
 because $\psi_j:X_j\rightarrow\CF{\GridLabel_j}{w_j}$
 and $\tilde\psi_j:X_j\rightarrow\CF{\GridLabel_j}{w}$ are uniform
 across dimension $k_j=1$.
Because $w_j\leq w$, it follows that
 $|\pi_i(\CF{\GridLabel_j}{w_j})|\geq|\pi_i(\CF{\GridLabel_j}{w})|$
 and
 $|\bpi_i(\CF{\GridLabel_j}{w_j})|\geq|\bpi_i(\CF{\GridLabel_j}{w})|$
 for $i\in [d]$.
This implies that $p\leq q$.

If $d=2$, then
 $R$ has $|X_j|\leq \ell'_{2}p\leq\ell'_2q$ edges.
Therefore,
 we can decompose $R$ into $\ell'_{1}$ edge-disjoint subgraphs
 $R_1,\ldots,R_{\ell'_{1}}$ so that
 $\bigcup_{i=1}^{\ell'_{1}} E(R_i)=E(R)$ and
 $|E(R_i)|\leq\lceil\ell'_{2}q/\ell'_{1}\rceil\leq q$
 for $i\in [\ell'_{1}]$.
Since $\ell''_{1}\geq\ell'_{1}$, $\bpi_{1}(R_i)$ can be routed
 with an edge-congestion at most $q$
 on the $1$-dimensional subgrid of $\GridLabel''$ induced by the nodes
 $\{v\in V(\GridLabel'')\mid\pi_{1}(v)=i\}$ for each
 $1\leq i\leq\ell'_{1}\leq\ell''_{1}$.
Thus, we can route $R$ on $\GridLabel''$ with an edge-congestion at most $2q$
 as 
 in the proof of Lemma~\ref{lm:GeneralRouting}.

If $d\geq 3$, then 
 since $\ell'_{2}=|\pi_{2}(\CF{\GridLabel_j}{w_j})|
\geq |\pi_{2}(\CF{\GridLabel_j}{w})|$,
 $\bpi_2(\bpi_1(R))$ is
 an $\ell'_{2}p$-$\ell'_{2}q$
 routing graph with node set $\bpi_{2}(\bpi_{1}(V(\GridLabel'')))$.
Using an edge-coloring
 described in the proof of Lemma~\ref{lm:GeneralRouting}, therefore,
 we can decompose $R$ into $\ell'_{1}$ edge-disjoint subgraphs
 $R_1,\ldots,R_{\ell'_{1}}$
 such that $\bpi_{2}(\bpi_{1}(R_i))$ is a $\max\{p,q\}$-$\max\{p,q\}$
 routing graph with node set $\bpi_{2}(\bpi_{1}(V(\GridLabel'')))$.
Since $\ell''_{1}\geq\ell'_{1}$,
 $\bpi_{1}(R_{i})$ can inductively  be routed
 with an edge-congestion at most $2\cdot\max\{p,q\}=2q$
 on the $(d-1)$-dimensional subgrid
 induced by the nodes $\{v\in V(\GridLabel'')\mid \pi_{1}(v)=i\}$
 for each $1\leq i\leq\ell'_{1}\leq\ell''_{1}$.
Thus, we can route $R$ on $\GridLabel''$ with
 an edge-congestion at most $2q$
 as in the proof of Lemma~\ref{lm:GeneralRouting}.
\end{proof}
\begin{lemma}
\label{lm:CompletionRouting}
The routing $\tilde\sigma$ in Step~{\theLabelSBERoutingEdges}
 imposes an edge-congestion at most
$2\cdot\max\{2D^w(n_j),\linebreak[0]D^w(n_j)+D^w(n)\}$
 on the channel of~$\CF\GridLabel{w}$.
\end{lemma}
\begin{proof}
Because $\tilde\psi_1$, $\tilde\psi_2$, and
 $\psi$
 are uniform across dimension $k$,
 it follows that
 $\lambda_{k}(\tilde\psi_j)
 =\lceil|X_j|/\CS{\GridLabel_j}{k}{w}\rceil
 \leq\lceil Cn_j^\alpha/\CS{\GridLabel_j}{k}{w}\rceil\leq D^w(n_j)$
 for $j=1,2$, and that
 $\lambda_{k}(\psi)
 =\lceil |X|/\CS{\GridLabel}{k}{w}\rceil
 \leq\lceil Cn^\alpha/\CS{\GridLabel}{k}{w}\rceil\leq D^w(n)$.
Therefore,
 if $R$ is the routing graph for $\tilde\sigma$ on the channel
 of $\CF\GridLabel{w}$,
 then
 $\bpi_{k}(R)$ has
 maximum outdegree at most
 $\lambda_{k}(\tilde\psi_1)+\lambda_{k}(\tilde\psi_2)\leq 2D^w(n_j) $ and
 maximum indegree at most
 $\lambda_{k}(\tilde\psi_2)+\lambda_{k}(\psi)\leq D^w(n_j)+D^w(n)$
 as shown in Fig.~\ref{fig:RoutingEdges}.
By Lemma~\ref{lm:GeneralRouting}, therefore,
 $\tilde\sigma$ has a desired edge-congestion.
\end{proof}
\begin{lemma}
\label{lm:D}
For $j=1,2$,
 $\max_{i\in [d]}\lceil Cn_j^\alpha/\CS{\GridLabel_j}{i}{w}\rceil$
 is $O(C)$ if $d>1/(1-\alpha)$, and $O(Cn_j^{\alpha-1+\frac 1 d})$ otherwise.
\end{lemma}
\begin{proof}
Suppose that $\GridLabel_j=\Grid{\ell^j_i}_{i\in [d]}$,
 $\lmax^j:=\max_{i\in [d]}\{\ell^j_i\}$, and
 $\lmin^j:=\min_{i\in [d]}\{\ell^j_i\}$.
We begin with bounds of $\lmax^j$ and
 $\lfloor(\ell_i-2)/2^w\rfloor$. 
Because $\lmin^j>2d$ by Lemma~\ref{lm:GridPartition},
 it follows that
\[
n_j\geq\prod_{i\in [d]}(\ell^j_i-2)
 =\prod_{i\in [d]}\left(1-\frac{2}{\ell^j_i}\right)\ell^j_i
 >\left(1-\frac{1}{d}\right)^d\left(\frac{\lmax^j}{\gM}\right)^d
 \geq\left(\frac{\lmax^j}{2\gM}\right)^d,
\]
 yielding
\begin{equation}
\label{eq:D_lh}
\lmax^j<
 2\gM n_j^{1/d}.
\end{equation}
It follows from the proof of Lemma~\ref{lm:ChannelNonEmpty} that
 $\lmin^j/2^w\geq 2$.
Therefore,
\begin{equation}
\label{eq:D_li}
\left\lfloor\frac{\ell^j_i-2}{2^w}\right\rfloor\geq\frac{\ell^j_i-2^w-1}{2^w}
 \geq\frac{\ell^j_i-\frac{\lmin^j}{2}-1}{2^w}
 \geq
\frac{\ell^j_i-2}{2^{w+1}}.
\end{equation}

If $k'$ is the direction of $\CF{\GridLabel_j} w$, then
 it follows from (\ref{eq:D_lh}) and (\ref{eq:D_li}) that
\[
\begin{split}
\min_{i\in [d]}\{\CS{\GridLabel_j} i w\}
 & \geq\left(\prod_{i\in\{1,2\}\setminus\{k'\}}
 \left\lfloor\frac{\ell^j_i-2}{2^w}\right\rfloor\right)
 \left(\prod_{i\in[d]\setminus\{1,2,k'\}}(\ell^j_i-2)\right)
 \geq\frac{1}{2^{2w+2}}\prod_{i\in[d]\setminus\{k'\}}(\ell^j_i-2)\\
 &=\frac{1}{2^{2w+2}}\prod_{i\in[d]\setminus\{k'\}}
 \left(1-\frac{2}{\ell^j_i}\right)\ell^j_i
 >
 \left(1-\frac{1}{d}\right)^{d-1}
 \frac{\prod_{i\in[d]\setminus\{k'\}}\ell^j_i}{2^{2w+2}}
>\frac{n_j/\lmax^j}{4e2^{2w}}
>\frac{n_j^{1-\frac{1}{d}}}{8e\gM2^{2w}}.
\end{split}
\]
Therefore,
 we have
 $\max_{i\in [d]}\lceil Cn_j^\alpha/\CS{\GridLabel_j}{i}{w}\rceil
 <\lceil 8e\gM 2^{2w}Cn_j^{\alpha-1+\frac{1}{d}}\rceil$.
If $w=0$, then the lemma is immediate.
If $w\geq 1$,  which implies
 $d>1/(1-\alpha)$,
 then it follows from inequalities in (\ref{eq:ChannelNonEmpty_w})
 and (\ref{eq:Ratio_njn}) that
\[
2^{2w}
 \leq\frac{(1-\beta)n^{\min\{1/d,1-\alpha-\frac 1 d\}}}{\gM}
 \leq\frac{(1-\beta)n^{1-\alpha-\frac 1 d}}{\gM}
 \leq\frac{(1-\beta)^{-1+2\alpha+\frac 2 d}n_j^{1-\alpha-\frac 1 d}}{\gM}.
\]
Therefore, we have
 $\max_{i\in [d]}\lceil Cn_j^\alpha/\CS{\GridLabel_j}{i}{w}\rceil
 <\lceil 8e(1-\beta)^{-1+2\alpha+\frac 2 d}C\rceil=O(C)$.
\end{proof}

Through an analysis similar to that of Lemma~\ref{lm:D},
 we can prove that
 $\lceil Cn^\alpha/\CS{\GridLabel}{k}{w}\rceil$ is $O(C)$ if $d>1/(1-\alpha)$,
 and $O(Cn^{\alpha-1+\frac 1 d})$ otherwise.
These upper bounds
 are independent of~$w$ and
 hold
 for any feasible guest and host graphs processed in each step of SBE
 unless the minimum side length of the host grid is at most $2d$.
In what follows, therefore,
 we simply write $D(\cdot)$ to denote $D^w(\cdot)$ and
 use these upper bounds of $D(\cdot)$ on the condition that the host grid
 has the minimum side length larger than $2d$.
\begin{lemma}
\label{lm:B}
The edge-congestion $B$ of the base embedding in
 Step~{\theLabelSBEBaseEmbedding} is $O(d\Delta+C)$.
\end{lemma}
\begin{proof}
The edge-congestion of $\rho$ constructed in the base embedding is at most
 $2\lceil\frac{\Delta}{2}\rceil \ell_h$
 by Lemmas \ref{lm:PermutationToEmbedding} and~\ref{lm:GeneralRouting}.
If $X\neq\emptyset$,
 then the edge-congestion of $\sigma$ is at most
 with an edge-congestion at most
 $2\lceil\gM(e|X|/\CS{\GridLabel} k w+\Delta)\rceil
\leq 2\lceil\gM(eD(n)+\Delta)\rceil$.
This bound is obtained from Lemma~\ref{lm:GeneralRouting2}
 and the fact that
 for the routing graph $R$ to be routed by $\sigma$,
 $\bpi_k(R)$ has
 maximum outdegree at most $e|X|/\CS{\GridLabel} k w+\Delta$
 by Lemma~\ref{lm:BaseEmbedding}
 and
 maximum indegree at most $\lceil |X|/\CS{\GridLabel} k w\rceil$
 as mentioned in Step~{\theLabelSBEBaseEmbedding}.
Because $n^{1/d}\leq \ell_h\leq 2\gM d$
 in the base embedding and $\min_{i\in [d]}\ell_i>2d$ as described
 in the proof of Lemma~\ref{lm:BaseEmbedding},
 we have $B\leq
 2(\lceil\frac{\Delta}{2}\rceil \ell_h+\lceil\gM(eD(n)+\Delta)\rceil)
=O(d\Delta+C)$ if
 $d>1/(1-\alpha)$,
 and $B=O(d\Delta+Cn^{\alpha-1+\frac{1}{d}})
 =O(d\Delta+C(2\gM d)^{d(\alpha-1)+1})=O(d\Delta+C)$ otherwise.
\end{proof}

We now estimate the total congestion of a fixed edge $r$ of $\GridLabel_0$.
If $r$ is contained in the channel of
 $\CF{\GridLabel_0}{w,w'}$ but
 of neither $\CF{\GridLabel_0}{w}$ nor $\CF{\GridLabel_0}{w'}$
 for a certain unique pair of $w>0$ and $w'>0$,
 then $r$ incurs a congestion less
 than that in the case that
 it is contained in the channel of either $\CF{\GridLabel_0}{w}$ or
 $\CF{\GridLabel_0}{w'}$.
This is because
 the channel of $\CF{\GridLabel_0}{w,w'}$ is used
 only in Step~{\theLabelSBERecursiveEmbedding}, while
 the channels
 of $\CF{\GridLabel_0}{w}$ and $\CF{\GridLabel_0}{w'}$
 are used in other steps as well. 
To analyze an upper bound of congestion of $r$,
 therefore, 
 it suffices to assume that
 $r$ is contained in the channel of $\CF{\GridLabel_0}{w_r}$ with $w_r\geq 1$
 uniquely determined by $r$, as well as in
 the channel of $\CF{\GridLabel_0}0$ and some base embeddings.
\begin{lemma}
\label{lm:P1}
The number $P_1$ of recursive calls of SBE that
 set $w\geq 1$ in Step~{\theLabelSBEChannelConfiguration},
 perform inductive steps (i.e., not a base embedding), and
 use a channel containing $r$ is 
 $O(d)$ if $d>2/(1-\alpha)$, 
 $O(1/(1-\alpha-\frac{1}{d}))$ if $1/(1-\alpha)<d\leq 2/(1-\alpha)$, 
 and $0$ otherwise.
\end{lemma}
\begin{proof}
Because SBE sets $w\geq 1$ only if $d>1/(1-\alpha)$,
 $P_1=0$ if $d\leq 1/(1-\alpha)$.
Assume $d>1/(1-\alpha)$.
In Step~{\theLabelSBERecursiveEmbedding}, channels configured
 in partitioned grids are edge-disjoint because the channels
 do not contain boundaries of the partitioned grids.
Therefore, there is a unique sequence of $P_1$ recursive calls that
 set $w\geq 1$,
 perform inductive steps, and
 use a channel containing $r$.
Two consecutive calls in the sequence are a parent and its child procedures.
Moreover, all but the first call (the ancestor of any other call) in
 the sequence set $w$ to $w_r\geq 1$, while
 the first call may set $w>w_r$ and use the channel associated with
 $w$ and $w_r$ in Step~{\theLabelSBERecursiveEmbedding}.
The number $n$ of nodes of the guest graph in the second call in the sequence 
 decreases to
\[
n'\leq\beta^{P_1-2}\left(n-\frac{1}{1-\beta}\right)+\frac{1}{1-\beta}\leq
 \beta^{P_1-2}n+\frac{1}{1-\beta}
\]
 at the last call in the sequence.
Because the last call performs inductive steps, it follows that
 $n'>2\gM d>2/(1-\beta)$.
Thus, we have $n'<\beta^{P_1-2}n+\frac{n'}2$, yielding $n'<2\beta^{P_1-2}n$.
Because the second and last calls
 set $w=w_r\geq 1$ in Step~{\theLabelSBEChannelConfiguration},
 it follows that
\[
\left\lfloor\frac{1}{2}\left(\left(1-\tilde\alpha-\frac{1}{d}\right)\log_2 n
-\log_2\frac\gM{1-\beta}\right)\right\rfloor
=\left\lfloor\frac{1}{2}\left(\left(1-\tilde\alpha-\frac{1}{d}\right)\log_2 n'
-\log_2\frac\gM{1-\beta}\right)\right\rfloor.
\]
Removing the floors,
\[
\frac{1}{2}\left(\left(1-\tilde\alpha-\frac{1}{d}\right)\log_2 n
-\log_2\frac\gM{1-\beta}\right)
<\frac{1}{2}\left(\left(1-\tilde\alpha-\frac{1}{d}\right)\log_2 n'
-\log_2\frac\gM{1-\beta}\right)+1.
\]
Combined with the upper bound of $n'$ obtained above,
\[
\log_2 n
<\log_2 n'+\frac{2}{1-\tilde\alpha-\frac{1}{d}}
<\log_2(2\beta^{P_1-2}n)+\frac{2}{1-\tilde\alpha-\frac{1}{d}},
\]
by which we obtain
\[
P_1<\frac{1+\frac{2}{1-\tilde\alpha-\frac{1}{d}}}{\log_2\beta^{-1}}+2.
\]
Because $\tilde\alpha=\max\{1-\frac 2 d,\alpha\}$,
 $P_1$ is $O(d)$ if $d>2/(1-\alpha)$, and $O(1/(1-\alpha-\frac 1 d))$ if
 $1/(1-\alpha)<d\leq 2/(1-\alpha)$.
\end{proof}
\begin{lemma}
\label{lm:P0}
The number $P_0$ of recursive calls of SBE that
 set $w=0$ in Step~{\theLabelSBEChannelConfiguration},
 perform inductive steps (i.e., not a base embedding), and
 use a channel containing $r$ is 
 $O(d)$ if $d>2/(1-\alpha)$, 
 $O(1/(1-\alpha-\frac{1}{d}))$ if $1/(1-\alpha)<d\leq 2/(1-\alpha)$, and
 at most $\log_{1/\beta} N$ otherwise.
\end{lemma}
\begin{proof}
By an argument similar to that in the proof of Lemma~\ref{lm:P1},
 there exists a unique sequence of $P_0$ recursive calls that
 set $w=0$,
 perform inductive steps, and
 use a channel containing $r$.
Moreover, $n$ in the first call of the sequence decreases to
 $n'$ with $2\gM d< n'<2\beta^{P_0-1}n$ at the last call in the sequence.
Therefore, it follows that
 $P_0 <\log_{1/\beta}n-\log_{1/\beta}(\gM d)+1
<\log_{1/\beta}n-\log_{1/\beta} 8+1<\log_{1/\beta}n$.
Because $n\leq N$ obviously, we have the lemma for the case
 $d\leq 1/(1-\alpha)$.
If $d>1/(1-\alpha)$, then $w=0$ implies that
\[
\left\lfloor\frac{1}{2}\left(\left(1-\tilde\alpha-\frac{1}{d}\right)\log_2 n
-\log_2\frac\gM{1-\beta}\right)\right\rfloor\leq 0.
\]
Removing the floor,
\[
\frac{1}{2}\left(\left(1-\tilde\alpha-\frac{1}{d}\right)\log_2 n
-\log_2\frac\gM{1-\beta}\right)<1,
\]
 by which we obtain
\[
\log_{1/\beta} n=\frac{\log_2 n}{\log_2\beta^{-1}}
<\frac{2+\log_2\frac{\gM}{1-\beta}}
{\left(1-\tilde\alpha-\frac 1 d\right)\log_2\beta^{-1}}.
\]
Because $\tilde\alpha=\max\{1-\frac 2 d,\alpha\}$,
 $P_0$ is $O(d)$ if $d>2/(1-\alpha)$, and $O(1/(1-\alpha-\frac 1 d))$ if
 $1/(1-\alpha)<d\leq 2/(1-\alpha)$.
\end{proof}
\begin{lemma}
\label{lm:EdgeCongestion}
The edge-congestion on $r$ is
 $O(dC+d^2\Delta)$ if $d>2/(1-\alpha)$,
 $O(C/(1-\alpha-\frac{1}{d})+d^2\Delta)$ if $1/(1-\alpha)<d\leq 2/(1-\alpha)$,
 and
 $O(C(N^{\alpha-1+\frac{1}{d}}+\log N)+d^2\Delta)$ otherwise.
\end{lemma}
\begin{proof}
The edge $r$ is congested by $\tilde\sigma_j$ with
 $j$ equal to either $1$ or $2$
 and $\tilde\sigma$
 in each recursive call performing inductive steps
 and using a channel containing $r$,
 and by base embeddings.
The congestion on $r$ imposed by $\tilde\sigma_j$ and $\tilde\sigma$
 in the $i$th recursive call in the sequence obtained by concatenating
 the sequences of recursive calls
 mentioned in Lemmas \ref{lm:P1} and~\ref{lm:P0}
 is at most
 $\max\{6D(N_{i+1}),4D(N_{i+1})+2D(N_i)\}\leq 6D(N_{i+1})+2D(N_i)$
 by Lemmas \ref{lm:PartialRouting} and~\ref{lm:CompletionRouting}, where
 $N_i$ is the number of nodes of a guest graph embedded in the $i$th recursive
 call in the concatenated sequence.
If $r$ is on the boundary of a host grid in some base embedding,
 then $r$ can be involved in at most $2(d-1)$ base embeddings in total.
Thus, the congestion on $r$ is at most
 $\sum_{i=1}^{P_1+P_0}(6D(N_{i+1})+4D(N_{i}))+2(d-1)B$.

By Lemmas \ref{lm:D}--\ref{lm:P0},
 this congestion is $O((P_0+P_1)C+d(d\Delta+C))=
 O(dC+d^2\Delta)$ if $d>2/(1-\alpha)$, and
 $O(C/(1-\alpha-\frac{1}{d})+d^2\Delta)$ if $1/(1-\alpha)<d\leq 2/(1-\alpha)$.
If $d\leq 1/(1-\alpha)$, then
 because
 $N_i\leq\beta N_{i-1}+1$, implying
 $N_i\leq\beta^{i-1}(N-\frac 1{1-\beta})+\frac{1}{1-\beta}=O(\beta^{i-1}N)$,
 we have
\[
\begin{split}
\sum_{i=1}^{P_1+P_0}(6D(N_{i+1})+4D(N_{i}))+2(d-1)B
 &\leq\sum_{i=1}^{\log_{1/\beta}N}
O\left(C\left(\beta^{i-1}N\right)^{\alpha-1+\frac{1}{d}}\right)
+O(d(d\Delta+C))\\
 &=O\left(C(N^{\alpha-1+\frac{1}{d}}+\log N)+d^2\Delta\right).
\end{split}
\]
\end{proof}

By Lemma~\ref{lm:GeneralRouting} and (\ref{eq:D_lh}),
 the dilation of SBE is at most
 $\sum_{i\geq 1}O(d(\beta^{i-1}N)^{1/d})
 =O(dN^{1/d})$.
Therefore, we have obtained Theorem~\ref{th:EmbedCore}.

\newcommand{\Gl}{G(\ell)}
\section{Lower Bound on Dilation with Minimum Edge-Congestion}
\label{sc:Dilation}
In this section, we demonstrate that
 minimizing edge-congestion may require a dilation
 of nearly the size of the host grid as stated in the following theorem:

\begin{theorem}
\label{th:Dilation}
There exists an $N$-node graph whose any embedding into
 an $N$-node $2$-dimensional grid with the edge-congestion $1$ has
 a dilation of $\Theta(N)$.
\end{theorem}

\begin{proof}
For an integer $\ell\geq 9$ with $\ell\bmod 4=1$,
 we define a guest graph $\Gl$
 obtained from $\Grid{\ell,\ell}$ by removing
 edges
\[
\begin{split}
&\{((i,j),(i,j+1))\mid 3\leq i\leq \ell-2,\ i\bmod 2=1,\ 3\leq j\leq\ell-3\}\\
\cup&
\{((i,3),(i+1,3))\mid 5\leq i\leq \ell-3,\ i\bmod 4\in\{1,2\}\}\\
\cup&
\{((i,\ell-2),(i+1,\ell-2))\mid 3\leq i\leq \ell-5,\ i\bmod 4\in\{3,0\}\}\\
\end{split}
\]
 and by adding an edge joining $e:=((3,3),(\ell-2,\ell-2))$.
We illustrate $G(13)$ in Fig.~\ref{fig:DilationGraph}.
The graph $\Gl$ can be embedded into
 $\Grid{\ell,\ell}$ with the edge-congestion $1$
 with an identity mapping for nodes and
 routing $e$ on the edges removed from $\Grid{\ell,\ell}$ to obtain $\Gl$.
This embedding clearly has a dilation of $\Theta(\ell^2)$.
We prove that
 if $\Gl$ can be embedded
 with the edge-congestion $1$
 into $\GridLabel:=\Grid{\ell_1,\ell_2}$ with
 $\ell_1\leq\ell\leq\ell_2$ and
 $\ell_1\ell_2=\ell^2$, then
 $\ell_1=\ell_2=\ell$ and
 such an embedding $\langle\phi,\rho\rangle$
 is unique within rotation and/or reflection.
Our proof is based on the following observation:
\begin{figure}
\begin{center}
\includegraphics[scale=1.0]{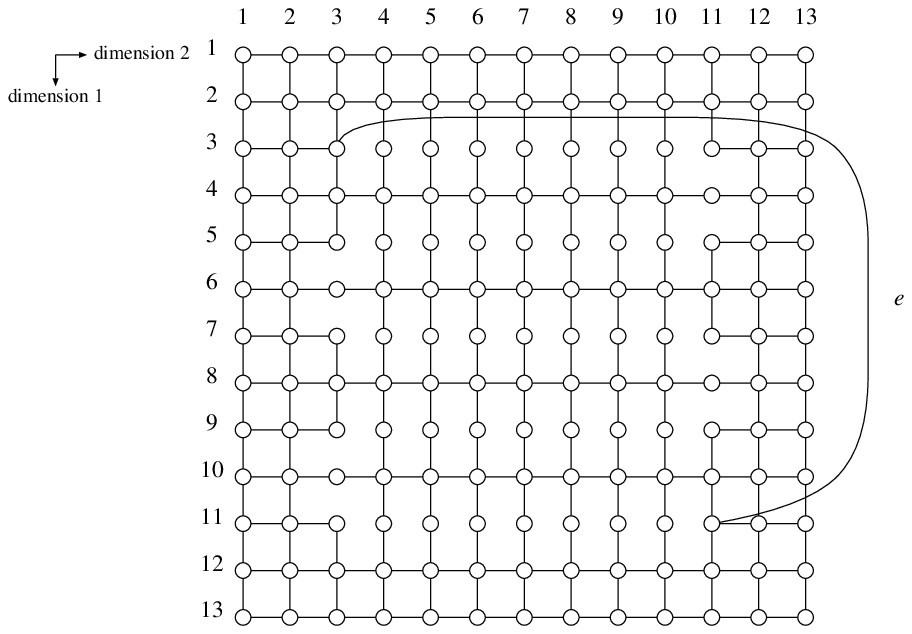}
\end{center}
\caption{%
$G(13)$.
}
\label{fig:DilationGraph}
\end{figure}
\begin{observation}
\label{ob:Dilation}
If $\rho$ maps $k$ edges of $\Gl$ on $k$ paths
 $h$ out of which ends at a node $v$ of $\GridLabel$,
 and the other $k-h$ of which pass through~$v$, then
\[
\deg_{\GridLabel}(v)\geq \deg_{\Gl}(\phi^{-1}(v))+2(k-h).
\]
\end{observation}
It should be noted that because
 $\Gl$ and $\GridLabel$ have exactly the same number of nodes,
 there exists a node $\phi^{-1}(v)$ of $\Gl$ for every node
 $v$ of $\GridLabel$.
We actually use this observation in different forms.

\begin{observation}
\label{ob:DilationDegree3}
If $\rho$ maps an edge of $\Gl$ on a path containing
 a node $v$ of $\GridLabel$ with degree $3$,
 then this edge is incident to $\phi^{-1}(v)$.
\end{observation}

\begin{observation}
\label{ob:DilationDegree4}
If $\rho$ maps two edges of $\Gl$ on two paths containing
 a node $v$ of $\GridLabel$ with degree $4$,
 then at least one of these edges is incident to $\phi^{-1}(v)$.
\end{observation}
Observations \ref{ob:DilationDegree3} and~\ref{ob:DilationDegree4} are implied
 by Observation~\ref{ob:Dilation}
 because
 $\Gl$ has no node with degree less than $2$,
 and therefore, for $k\geq\lfloor\deg_{\GridLabel}(v)/2\rfloor$,
\[
h\geq\frac{\deg_{\Gl}(\phi^{-1}(v))} 2+k
-\frac{\deg_{\GridLabel}(v)} 2
\geq 1+\left\lfloor\frac{\deg_{\GridLabel}(v)} 2\right\rfloor
-\frac{\deg_{\GridLabel}(v)} 2
 \geq\frac 1 2,
\]
 implying $h\geq 1$.

We first identify nodes of $\Gl$ mapped onto the boundary of $\GridLabel$.
Because $\Gl$ has no node with degree less than $2$,
 the node $\phi^{-1}((1,1))$ has degree~$2$.
Two edges of $\Gl$ incident to $\phi^{-1}((1,1))$ must be routed on nodes
 $(1,2)$ and $(2,1)$ of $\GridLabel$ with degree~$3$.
By Observation~\ref{ob:DilationDegree3}, therefore, 
 these edges are incident to $\phi^{-1}((1,2))$ and $\phi^{-1}((2,1))$.
Because only four corner nodes of $\Gl$,
 i.e., $(1,1)$, $(1,\ell)$, $(\ell,1)$, and $(\ell,\ell)$
 have degree $2$ and are incident to a node with degree $3$,
 we may assume without loss of generality that $\phi^{-1}((1,1))=(1,1)$,
 $\phi^{-1}((1,2))=(1,2)$, and $\phi^{-1}((2,1))=(2,1)$.
Repeating a similar argument, we can identify
 $\phi^{-1}((i,1))=(i,1)$ for $3\leq i\leq \ell_1$.
This implies that if $\ell_1<\ell$, then
 $\deg_{\GridLabel}((\ell_1,1))=2$ and $\deg_{\Gl}((\ell_1,1))=3$,
 yielding an edge-congestion more than~$1$.
Hence, we obtain $\ell_1=\ell_2=\ell$.
As a consequence,
 $\phi((i,1))=(i,1)$, $\phi((i,\ell))=(i,\ell)$,
 $\phi((1,i))=(1,i)$, and $\phi((\ell,i))=(\ell,i)$
 for $1\leq i\leq \ell$.

We then identify nodes of $\Gl$ mapped onto nodes on one row and one column
 inside
 the boundary of $\GridLabel$.
Because $\phi((1,2))=(1,2)$ and $\phi((2,1))=(2,1)$,
 two edges of $\Gl$ incident to $(1,2)$ and $(2,1)$ must be routed on
 the node $(2,2)$ of $\GridLabel$.
By Observation~\ref{ob:DilationDegree4}, therefore, 
 at least one of these two edges of $\Gl$ is incident
 to $\phi^{-1}((2,2))$.
Thus, we can identify $\phi^{-1}((2,2))=(2,2)$ because
 all the other nodes of $\Gl$
 adjacent to $(1,2)$ or $(2,1)$, i.e, $(1,1)$, $(1,3)$, and $(3,1)$
 have already been identified to be mapped to other positions.
Repeating a similar argument, we obtain
 $\phi((i,2))=(i,2)$, $\phi((i,\ell-1))=(i,\ell-1)$,
 $\phi((2,i))=(2,i)$, and $\phi((\ell-1,i))=(\ell-1,i)$
 for $2\leq i\leq \ell-1$.

Because $\phi((2,3))=(2,3)$ and $\phi((3,2))=(3,2)$,
 we can identify $\phi((3,3))=(3,3)$
 as done for $\phi((2,2))=(2,2)$,
 and similarly,
 $\phi((i,3))=(i,3)$ for $i\in\{3,4,5,\ell-2\}$ and
 $\phi((i,\ell-2))=(i,\ell-2)$ for $i\in\{3,\ell-4,\ell-3,\ell-2\}$.

Now we identify the routing of $e$.
Two edges of $\Gl$ incident to $\phi^{-1}((3,3))=(3,3)$ and
 $\phi^{-1}((2,4))=(2,4)$ must be routed on
 the node $(3,4)$ of $\GridLabel$.
By Observation~\ref{ob:DilationDegree4}, therefore, 
 we can identify $\phi^{-1}((3,4))=(3,4)$
 and $\rho(e)$ passing through $(3,4)$
 because
 all the other nodes of $\Gl$
 adjacent to $(3,3)$ or $(2,4)$, including $(\ell-2,\ell-2)$,
 have already been identified to be mapped to other positions.
With this fact,
 either $e$ or an edge of $\Gl$ incident to $\phi^{-1}((3,4))=(3,4)$, and
 an edge of $\Gl$ incident to $\phi^{-1}((4,3))=(4,3)$ must be routed on
 the node $(4,4)$ of $\GridLabel$.
By Observation~\ref{ob:DilationDegree4} again,
 we can identify $\phi^{-1}((4,4))=(4,4)$
 because
 all the other nodes of $\Gl$
 adjacent to $(3,3)$, $(3,4)$, or $(4,3)$
 have already been identified to be mapped to other positions.
This implies that $\rho(e)$ passes through $(3,4)$ toward $(3,5)$.
Repeating a similar argument, we obtain
 $\phi((3,i))=(3,i)$,
 $\phi((4,i))=(4,i)$, and
 $\rho(e)$ passes through $(3,i)$ toward $(3,i+1)$
 for $4\leq i\leq \ell-3$.

The path $\rho(e)$ passes through $(3,\ell-2)$ toward $(4,\ell-2)$ because
 it cannot go toward other directions.
Then,
 $\rho(e)$
 passes through $(4,\ell-2)$ and $(5,\ell-2)$ toward $(5,\ell-3)$
 with fixing
 $\phi^{-1}((4,\ell-2))=(4,\ell-2)$,
 $\phi^{-1}((5,\ell-2))=(5,\ell-2)$, and
 $\phi^{-1}((6,\ell-2))=(6,\ell-2)$
 as similarly discussed above.
We can also identify $\phi^{-1}((7,\ell-2))=(7,\ell-2)$
 since 
 $\phi^{-1}((6,\ell-2))=(6,\ell-2)$ and $\phi^{-1}((7,\ell-1))=(7,\ell-1)$.

At this point we have obtained
 the situation for $\rho(e)$ leaving from $(5,\ell-2)$ toward $(5,\ell-3)$,
 together with identified nodes of $\Gl$ mapped onto the $4$th row,
 $(i,\ell-2)$ for $i\in\{5,6,7\}$, and onto $(5,3)$.
Continuing this process until $\rho(e)$ arrives at $(\ell-2,\ell-2)$,
 we conclude that the embedding $\langle\phi,\rho\rangle$ is unique.
\end{proof}

\section{Concluding Remarks}
An open question is to improve the approximation ratio for $d\leq 1/(1-\alpha)$.
A main defect of SBE in approximation for $d\leq 1/(1-\alpha)$
 is 
 the use of an edge of the host grid
 in $\Theta(\log N)$ recursive steps,
 yielding
 a gap of $\Theta(\log N)$ factor to the optimal edge-congestion
 in the worst case.
Another open question is to improve the dilation.
In this connection,
 the author suspects that there is a general trade-off between
 edge-congestion and dilation, such as
 existence of guest graphs whose any embedding into a grid
 does not allow
 constant ratio approximation for both dilation and edge-congestion.

An analogous fact to Theorem~\ref{th:Dilation} for hypercubes
 can also be proved using the existence of an induced path
 of length $\Theta(N)$ in an $N$-node hypercube \cite{AK88}.



\end{document}